\documentclass[11pt]{article}
\usepackage{mathtools, amsthm, amsfonts, amssymb,bm,amsmath}
\usepackage[round]{natbib}
\usepackage{fullpage}
\usepackage{authblk}
\usepackage{setspace}
\usepackage[unicode]{hyperref}
\usepackage{pgfplots}
\usepackage{multicol}
\usepackage{thmtools}
\usepackage{subcaption}
\usepackage{algorithm}
\usepackage{algpseudocode}
\usepgfplotslibrary{fillbetween}

\newcommand{\Break}{\State \textbf{break} }

\theoremstyle{plain}
\declaretheorem[name=Theorem]{thm}
\newtheorem{prop}{Proposition}

\newtheorem{lem}{Lemma}

\newtheorem{assumption}{Assumption}
\newtheorem{defn}{Definition}

\newcommand{\bz}{{\boldsymbol{z}}}
\newcommand{\bN}{{\boldsymbol{N}}}
\newcommand{\bx}{{\boldsymbol{x}}}
\newcommand{\by}{{\boldsymbol{y}}}
\newcommand{\ba}{{\boldsymbol{a}}}
\newcommand{\bv}{{\boldsymbol{v}}}
\newcommand{\br}{{\boldsymbol{r}}}
\newcommand{\balpha}{{\boldsymbol{\alpha}}}
\newcommand{\bell}{{\boldsymbol{\ell}}}
\newcommand{\diag}{\operatorname{diag}}
\renewcommand{\Re}{\operatorname{Re}}
\renewcommand{\Im}{\operatorname{Im}}
\newcommand{\BS}{\mathcal{B}}
\newcommand{\FS}{\mathcal{F}}
\newcommand{\AS}{\mathcal{A}}
\newcommand{\LN}{\mathcal{N}}
\newcommand{\Lx}{\mathcal{X}}
\newcommand{\im}{\mathrm i}

\newcommand{\DLBshort}{DGD-LB}
\newcommand{\DLBlong}{Distributed Gradient Descent Load Balancing}

\begin{document}

\title{Load Balancing with Network Latencies\\via Distributed Gradient Descent}

\author[1,2]{Santiago R. Balseiro\thanks{srb2155@columbia.edu}}
\author[2]{Vahab S. Mirrokni\thanks{mirrokni@google.com}}
\author[2]{Bartek Wydrowski\thanks{bwydrowski@google.com}}
\affil[1]{Columbia University}
\affil[2]{Google Research}

\date{\today}

\maketitle

\begin{abstract}
Motivated by the growing demand for serving large language model inference requests, we study  distributed load balancing for global serving systems with network latencies. We consider a fluid model in which continuous flows of requests arrive at different frontends and need to be routed to distant backends for processing whose processing rates are workload dependent. Network latencies can lead to long travel times for requests and delayed feedback from backends. The objective is to minimize the average latency of requests, composed of the network latency and the serving latency at the backends.

We introduce {\DLBlong} ({\DLBshort}), a probabilistic routing algorithm in which each frontend adjusts the routing probabilities dynamically using gradient descent. Our algorithm is distributed: there is no coordination between frontends, except by observing the delayed impact other frontends have on shared backends. The algorithm uses an approximate gradient that measures the marginal impact of an additional request evaluated at a delayed system state. Equilibrium points of our algorithm minimize the centralized optimal average latencies, and we provide a novel local stability analysis showing that our algorithm converges to an optimal solution when started sufficiently close to that point. Moreover, we present sufficient conditions on the step-size of gradient descent that guarantee convergence in the presence of network latencies. Numerical experiments show that our algorithm is globally stable and optimal, confirm our stability conditions are nearly tight, and demonstrate that {\DLBshort} can lead to substantial gains relative to other load balancers studied in the literature when network latencies are large.
\end{abstract}

\setstretch{1.5}

\section{Introduction}

Load balancing is essential in today's distributed computing systems, especially with the rise of large language models (LLMs).  LLMs require significant and expensive computational resources for inference requests, and efficiently distributing these requests across multiple servers is crucial for scalability and performance.  Without effective load balancing, some servers could be overloaded while others sit idle, leading to increased latency, reduced throughput, and potential system instability. 

The goal of this paper is to present a distributed algorithm for load balancing requests across backends in the presence of network latencies when backends' processing rates are non-linear functions of their workloads, effectively reflecting the inference capacity. This combination of features distinguishes our work from existing algorithms, which either lack distributed operation~\citep{bassamboo2006design}, are suboptimal in the presence of non-linear processing rates~\citep{winston1977optimality,weber1978optimal,tassiulas1990stability}, or fail to account for network latencies~\citep{zhang2024optimal}.

\subsection{Main Contributions}

We consider a bipartite network with frontends on one side and backends on the other. The network represents a global load balancing system in which frontends are the entry point for user requests and each backend is a different data center with multiple servers. There is a continuous flow of requests arriving at the frontends that need to be instantly routed to the distant backends for processing. Network latencies between frontends and backends lead to extended travel times for requests and delayed feedback from backends. Due to so-called locality constraints, a frontend's requests can only be processed by a specific subset of backends that host the necessary data or machine learning models for those tasks~\citep{weng2020optimal,rutten2023mean}.

The processing rates of these backends are concave, increasing functions of their workloads. This aims to capture that backends can have multiple servers and the complex nature of large language model inference tasks conducted in each server, which can have multiple stages (prefill, decoding, etc.) with different levels of parallelism, resource bottlenecks, caching, and batching optimizations~\citep{kwon2023efficient, wan2024efficient}.

The objective is to minimize the average requests traveling in the network and in service at the backends. By Little's Law, this is equivalent to minimizing the average latency experienced by requests, which comprises both the network latency and the serving latency at the backends. We benchmark our algorithm against the optimal static routing, which is the centrally optimal solution, and a universal bound on the performance of every policy.

We develop a novel load balancer for bipartite networks with network latencies called {\DLBlong} ({\DLBshort}), which is a probabilistic routing algorithm where each frontend dynamically adjusts routing probabilities independently using gradient descent. Namely, each frontend maintains a vector of routing probabilities and arriving requests are routed to a frontend according to these probabilities. Frontends compute ``approximate gradients'' that capture the marginal impact of an additional request, evaluated at a delayed system state, and update the routing probabilities using projected gradient descent. The projection operation guarantees that routing decisions lie in the probability simplex.

Our algorithm operates in a distributed manner, eliminating the need for coordination between frontends except through observing the delayed impact on shared backends. This eliminates unnecessary communicating overhead between frontends and leads to a robust algorithm with no single point of failure. In addition, {\DLBshort} is a dynamic algorithm that can quickly react to changes in the frontends' arrival rates, network topology, or backends processing capacity. Notably, our algorithm does not need to know the arrival rates at the frontends.

This research makes several key technical contributions:

\begin{itemize}
   
   
\item \emph{Optimality of Equilibrium Points}: We  prove that all equilibrium points of {\DLBshort} are optimal solutions to the static routing problem. This guarantees that the algorithm can only converge to a solution that minimizes the average latency of requests. 
   
\item \emph{Local Stability Analysis}: The dynamics of {\DLBshort} can be modeled as delayed differential equation with a right-hand side that is discontinuous due to gradient descent's projection operation. We conduct a novel local stability analysis of the algorithm, demonstrating its convergence to an optimal solution when initialized sufficiently close to that point. This analysis involves linearizing the dynamics around an equilibrium point and carefully considering the impact of the non-linear projection operator. 
   
\item \emph{Convergence Conditions}: We establish sufficient conditions on the step-size of gradient descent that ensure convergence for general networks in the presence of heterogeneous network latencies. These conditions provide practical guidelines for tuning the algorithm's parameters to achieve stable and efficient performance. 
   
\item \emph{Numerical Validation}: We perform numerical experiments to confirm the stability and optimality of our algorithm. These experiments also validate the tightness of our stability conditions for networks with one frontend and their sufficiency for general networks. Experiments illustrate that inappropriate tuning of step-sizes can lead to oscillatory behavior, which is typical of control systems with delay. Our numerical experiments confirm that our algorithm is globally stable and, in the presence of network latencies, drastically outperforms other popular load balancing algorithms that route based on the workload or latencies of backends. In the presence of delays, dynamic policies that do not incorporate delays can lead to oscillatory behavior and poor performance.
\end{itemize}


The stability analysis of our algorithm is complex due to the interplay of (i) request dynamics: requests arrive at frontends and depart from backends at varying rates, creating dynamic workloads that need to be balanced, (ii) network latencies: frontends received delayed feedback from backends and requests take time to travel from frontends to backends, and (iii) arbitrary network topologies: the algorithm must operate effectively in arbitrary network topologies, including those with complex connectivity patterns and heterogeneous latencies. 

Our analysis draws inspiration from the local stability analysis of congestion control algorithms for TCP, which have been successfully deployed in the real world~\citep{brakmo1995tcp}. {\DLBshort} shares similarities with primal-dual TCP algorithms, particularly in maintaining states in both the sources and links~\citep{srikant2004mathematics,low2022analytical}. However, the presence of a complex projection operation and the need to handle arbitrary network topologies necessitate the development of novel analytical techniques. These techniques, which include using Lyapunov's direct method to handle the projection operation, a uniformization technique for handling multiple frontends, and a geometric analysis of the eigenvalues of the loop transfer function, may be of independent interest in the broader field of distributed control systems. While this work provides a foundation for analyzing global stability, its complete characterization remains an open problem. Proving global stability for distributed control algorithms with delay is notoriously difficult~\citep{massoulie2002stability,paganini2005congestion}.

\subsection{Other Related Work}

There is an extensive literature on load balancing in computing system. We refer the reader to \citet{al2012survey} and \citet{der2022scalable} for recent surveys of load balancing algorithms in cloud computing and scalable systems.

Load balancing algorithms can be categorized as either distributed or centralized, and as either static or dynamic. In this paper, we consider distributed load balancing algorithms that have the advantage over centralized algorithms (such as \citealt{bassamboo2006design}) of being more robust, as there is no single point of failure. Static load balancing algorithms such as weighted round robin~\citep{hajek1985extremal} and probabilistic assignment~\citep{combe1994optimization} require prior knowledge of the system and cannot adapt to changes in the system. Dynamic algorithms like {\DLBshort}, in contrast, can react better to changes in the system and lead to performance improvements. 

Dynamic policies such Join-the-Shortest-Queue (JSQ)~\citep{winston1977optimality,weber1978optimal,gupta2007analysis}, Join-Idle-Queue (JIQ)~\citep{lu2011join}, and MaxPressure~\citep{tassiulas1990stability} are known to stabilize systems and attain near-optimal latencies when backends are homogeneous. These policies, however, do not incorporate network latencies and can be suboptimal when backends are heterogeneous with load-dependent processing rates. Moreover, policies such as JSQ and JIQ cannot be implement in our setting with workload-dependent processing rates as there is no clear notion of queues and idle servers. We refer to \citet{van2018scalable} for a comparison of many of these policies. 

\citet{weng2020optimal} studies load balancing in bipartite graphs with locality constraint with heterogeneous backends with one exponential server each. They evaluate two load balancing policies, Join-the-Fastest-of-the-Shortest-Queue (JFSQ) and Join-the-Fastest-of-the-Idle-Queue (JFIQ), demonstrating their asymptotic optimality in minimizing mean response time when the number of servers is large under a ``well-connected'' graph condition. \citet{rutten2023mean} also consider bipartite graphs with locality constraint and homogeneous servers. They study a power-of-$d$ choices load-balancing algorithm and show that, when the network size is large, arbitrary graphs with diverging degrees behave like a fully flexible system on a complete bipartite graph. \citet{horvath2019mean} study a single-frontend network in which backends are heterogeneous and have workload-dependent service rate functions as the one considered in our paper. They propose a threshold policy, which they analyze in a mean field regime, and propose an optimization approach to determine the thresholds that minimize response times.

In the presence of feedback delays, \citet{litvak2003routing} and \citet{mehdian2017join} study the performance of Round Robin and Join-the-Shortest-Queue, respectively. Interestingly, \citet{mehdian2017join} shows that Join-the-Shortest-Queue can lead to poor performance when feedback delays are long.

The most related paper is \citet{zhang2024optimal}, which presents a distributed load balancing algorithm called Greatest Marginal Service Rate (GMSR) that can handle non-linear processing rate functions. Their algorithm and analysis, however, ignores network latencies, which can be substantial in global systems. GMSR is a discrete or bang-bang control algorithm in which each request is routed to the connected backend with the smallest gradient. Their algorithm can be interpreted as a special case of our gradient-descent based algorithm when the step-sizes goes to infinity. Our analysis and simulations show that GMSR is unstable and suboptimal in the presence of network latencies.
%


\section{Model}

Consider a bipartite network or graph $\mathcal G = (\FS, \BS, \AS)$ where $\FS$ is the set of frontends, $\BS$ is the set of backends, and $\AS$ is the set of arcs.  We denote by $\BS_i = \{ j \in \BS : (i,j) \in \AS\}$ and $\FS_j = \{ i \in \FS : (i,j) \in \AS\}$ the neighbors of a frontend and backend, respectively. 

We study a fluid model in which requests arrive to frontend $i \in \FS$ at rate $\lambda_i > 0$, with units of requests per second. Fluid models are used extensively in the queuing literature~\citep{chen2001fundamentals} and congestion control literature~\citep{srikant2004mathematics} as they approximate well the behavior of large systems and lead to more analytically tractable models.

We denote by $N_j(t)$ the amount of requests or workload in backend $j \in \BS$ at time $t$. Backend $j \in \BS$ has a processing rate function  is $\ell_j : \mathbb R_+ \rightarrow \mathbb R_+$ that determines the processing rate $\ell_j(N_j(t))$ in requests per second as a function of its workload. In global systems, backends, such as data centers with numerous servers, often employ internal load balancing algorithms for routing incoming requests and servers, in turn, can process multiple requests in parallel, sharing resource between requests. The processing rate functions are designed to model the overall processing capacity of a backend, effectively abstracting away its internal complexities. This approach of using workload-dependent service rates to model resource sharing servers has been explored in similar contexts, as noted by \citet{horvath2019mean} and \citet{zhang2024optimal}.

The travel time or latency from frontend $i \in \FS$ to backend $j \in \BS$ is $\tau_{ij}>0$. The units are seconds. We assume links have sufficient bandwidth so that network latencies are constant and load independent. This is a reasonable assumption for LLM-based queries for which prompts and responses can only require a few kilobytes of data.  In addition, denote by $N_{ij}(t)$ the amount of requests traveling from frontend $i \in \FS$ to backend $j \in \BS$ at time $t$.

We conduct our analysis under the following assumption on the processing rate functions.

\begin{assumption}\label{assume:processing-rate} The processing rate functions $\ell_j$ are strictly increasing, concave, and twice differentiable.
\end{assumption}

The monotonicity of the processing rate function is a natural assumption for work-conserving backends that never idle when requests are available for processing. Recall that $\ell_j(N_j)$ measures the total processing rate of backend $j$ when $N_j$ requests are available---not the processing rate of an individual request. Therefore, for higher workloads the total processing rate should never decrease. Concavity of the processing rate functions follows because backends usually exhibit decreasing returns to scale because of congestion of resources and increased contention. In the context of LLMs, throughput has been observed to increase with more requests, even for a single GPU, as computational overheads are amortized, improving efficiency. However, as explained in \citet{kwon2023efficient}, this gain diminishes as the system becomes constrained by GPU memory limits needed for the key-value cache of each request, leading to the concave throughput relationship. Twice differentiability is a technical assumption made to simplify the analysis. We suspect our results hold without this smoothness assumption. We refer the reader to \citet{zhang2024optimal} for additional motivation for this assumption.

\paragraph{Routing policies} We consider probabilistic or proportional routing policies that at time $t$ route a fraction $x_{ij}(t)$ of the incoming flow from frontend $i \in \FS$ to backend $j \in \BS$. The set of feasible routing probabilities for frontend $i$ is a probability simplex in which we restrict to zero the flow of arcs not in $\AS$, i.e.,
\[
    \Delta_i = \left\{\bx \in \mathbb{R}^{|\BS|}: x_j \geq 0, x_j = 0, \forall j \notin \BS_i, \sum\nolimits_{j\in\BS_i} x_j = 1 \right\}\,.
\]
Therefore, at every time $t$ we should have that $\bx_i(t) = (x_{ij}(t))_{j \in \BS} \in \Delta_i$.

The dynamics for $t > 0$ are given as follows
\begin{equation}\label{eq:dynamics-workloads}
\begin{split}
     \frac {d} {dt} N_j(t) &= \sum_{i\in\FS_j} \lambda_i x_{ij}(t - \tau_{ij}) - \ell_j (N_j(t))\,, \quad \forall j \in \BS\,, \\
    \frac {d} {dt} N_{ij}(t) &= \lambda_i \left( x_{ij}(t) - x_{ij}(t - \tau_{ij})\right)\,, \quad\forall (i,j) \in \AS\,.
\end{split}
\end{equation}
The initial conditions for times $t \le 0 $ are exogenously given. At each backend, there is an inflow of requests that originates from its connected frontends with a delay that takes into account travel times and an outflow based on the processing rate of the backend. The dynamics for the amount of requests in a link connecting a frontend and a backend leverage that requests have deterministic travel times. The objective is to minimize the long-run average amount of requests in the system, which by Little's Law, is equivalent to minimizing the average latency of requests. That is, the performance of a routing policy as measures as follows
\[
    \operatorname{ALG} = \lim_{T \rightarrow \infty} \frac 1 T \int_{0}^T \left( \sum_{j \in \BS} N_j(t) + \sum_{(i,j) \in \AS} N_{ij}(t) \right) dt\,.
\]
To simplify the exposition, our analysis does not include the network latency of requests traveling back from backends to frontends, which is nevertheless easy to incorporate.

\subsection{Optimal Static Routing}

We benchmark our algorithms against the optimal static routing policy that minimizes the average amount of requests in the system. 

The \emph{optimal static routing} problem can be formulated as
\begin{equation}\label{eq:opt}
\begin{split}
    \operatorname{OPT} = \min_{\bN,\bx} &\quad \sum_{j \in \BS} N_j + \sum_{(i,j) \in \AS} \lambda_i x_{ij} \tau_{ij} \\
    \text{s.t.}
    & \quad \sum_{i \in \FS_j} \lambda_i x_{ij} = \ell_j(N_j)\,, \forall j \in \BS\,,\\
    & \quad  \sum_{j \in \BS_i} x_{ij} = 1\,, \quad \forall i \in \FS\,,\\
    & \quad x_{ij} \geq 0\,, \quad \forall (i,j) \in \AS\,.
\end{split}
\end{equation}
The objective has two terms. The first term captures the average amount of requests in the backend and the second term captures the average amount of requests traveling from each frontend to a backend, which, by Little's Law, is given by the product of the flow and the travel time.  The first constraint says that, at equilibrium, the processing rate of each backend $j \in \BS$ should be equal to the flow of incoming requests. If inflows exceed outflows at a backend, then the workload of a backend could eventually explode. The second constraint guarantees that all requests are routed to a backend.

The objective value of the optimal static routing problem provides a lower bound to the performance of any feasible policy. To provide this result, we restrict attention to \emph{time-average convergent policies}, i.e., policies for which the limits $\int_0^T N_j(t) dt / T$ and $\int_0^T N_{ij}(t) dt/T$ as $T\rightarrow\infty$ exist and are finite, as otherwise their performance is not well defined. We prove this result by using Jensen's inequality to establish that the time-average workloads and routing probabilities are feasible for a relaxed version of \eqref{eq:opt} in which the inflow to each backend is at most its outflow. All proofs are located in the appendix.

\begin{lem}\label{lem:opt-bound}
    For every time-average convergent policy we have that $\operatorname{ALG} \ge \operatorname{OPT}$.
\end{lem}

The following result provides the first-order  conditions for the optimal static routing problem, which play a key role in the design of our algorithm and the ensuing analysis. We prove the result by invoking the Karush-Kuhn-Tucker conditions of the static routing problem.

\begin{lem}[First-Order Conditions]\label{lem:foc}
At an optimal solution $(\bN^*, \bx^*)$, for each frontend  $i \in \FS$ there exists a constant $c_i > 0$ such that for backend $j \in \BS_i$ we have $1/\ell_j'(N_j^*) + \tau_{ij} \ge c_i$ with equality holding if $x_{ij}^* >0$.
\end{lem}

The first-order conditions state that, for each frontend $i$, pushing an additional unit of flow through any outgoing arc $(i,j) \in \AS$ that has positive flow at an optimal solution must lead to the same marginal amount of requests in the system. Otherwise, we could improve the solution by increasing the flow of the arc with the lowest marginal impact on the objective. 
This ``marginal impact'' of an additional unit of flow, which we denote by $c_i$, is equal to the quantity $1/\ell_j'(N_j^*) + \tau_{ij}$. Intuitively, $1/\ell_j'(N_j^*)$ captures the marginal amount of requests in backend $j$ resulting from an additional unit of flow, while $\tau_{ij}$ captures the marginal amount of requests traveling between $i$ and $j$ when an additional unit of flow travels through this arc. For arcs that have no flow at an optimal solution, the marginal impact on the objective is naturally greater than $c_i$ as pushing flow through these is suboptimal.

The constants $c_i > 0$ associated to each frontend $i \in \FS$ are the Lagrange multipliers of the frontend flow balance constraint $\sum_{j \in \BS_i} \lambda_i x_{ij} = \lambda_i$ and measures the marginal impact in the objective, whose units are requests, of increasing the arrival rate of frontend $i$ by one request per second. The units of $c_i$ are in seconds.

\section{{\DLBlong} Algorithm}

Our goal is to design an algorithm in which  frontends make independent decisions. Our algorithm only requires communication between frontends and backends, which is essential to route requests but no additional coordination between frontends, except by observing the impact other frontends have on the backends. We refer to our algorithm as {\DLBlong} or, hereafter, {\DLBshort}.

{\DLBshort} is a decentralized first-order algorithm in which at time $t$ each frontend $i \in \FS$ computes an ``approximate gradient'' $\boldsymbol g_i(t) \in \mathbb R^{|\BS|}$ of the objective and then updates the routing probabilities $\bx_{i}(t) \in \mathbb R^{|\BS|}$ by moving in a direction opposite to the gradient using projected gradient descent to guarantee feasibility of the routing probabilities. Motivated by the first-order conditions in Lemma~\ref{lem:foc}, the gradient of with respect to $x_{ij}$ at time $t$ for $(i,j) \in \AS$ is approximated by
\[
    g_{ij}(t) = \frac 1 {\ell_j'(N_j(t-\tau_{ij}))} + \tau_{ij}\,,
\]
and $g_{ij}(t) = \infty$ for $(i,j) \not\in \AS$. The latter expression accounts for the feedback delay from backend $j$ to frontend $i$. We provide some intuition for this choice in Section~\ref{sec:motivation}. Note that backends can communicate to frontends either their gradients or their workloads and let frontends compute gradients. The former is advantageous because frontends do not need to know the processing rate functions.

Before introducing the algorithm, we need to specify the projection to the probability simplex for each frontend $i\in\FS$. First, let $T_{\Delta_i}(\bx_i)$ be the tangent cone of $\Delta_i$ at $\bx_i$, which is given by
\[
 T_{\Delta_i}(\bx_i) = \left\{\bv \in \mathbb{R}^{|\BS|}: \sum_{j\in \BS_i} v_j = 0,  v_j \geq 0  \text{ if } x_{ij}  = 0, v_j = 0 \text{ for all } j \not\in \BS_i \right\}\,.
 \]
Intuitively, the tangent cone captures directions along which the frontend can update the routing probabilities while maintaining feasibility. The components of a feasible direction $\bv \in T_{\Delta_i}(\bx_i)$ should sum up to zero to satisfy the constraint that probabilities sum up to one. Moreover, for backends whose probabilities are at zero, the corresponding component of the direction should be non-negative to preserve the non-negativity constraint. Second, given a vector $\boldsymbol{z} \in \mathbb{R}^{|\BS|}$, we define the Euclidean projection to a set $\mathcal C$ as the point in the set which is closest to $\boldsymbol{z}$ with respect to the Euclidean norm, and is given by
\[
\Pi_{\mathcal C}(\boldsymbol{z}) = \arg\min_{\bv \in \mathcal C} \|\bv - \boldsymbol{z}\|_2\,.
\]
For a closed and convex set $\mathcal C$, the projection always exists and is unique.
If the set $\mathcal C$ has a constraint $v_j = 0$, we ignore the corresponding component $(v_j - z_j)^2$ in the norm. We introduce this convention to allow gradients to take value of infinity for arcs that are not in the network.

The routing probabilities of our algorithm are updated as follows
\begin{equation}\label{eq:dynamics-delay-general}
\begin{split}
    \frac {d} {dt} \bx_{i}(t) &= \Pi_{T_{\Delta_i}(\bx_i(t))} \left(-\eta_i \boldsymbol{g}_i(t) \right)\,,
\end{split}
\end{equation}
where $\boldsymbol{g}_i(t) = (g_{ij}(t))_{j\in\BS}$ and $\eta_i > 0$ is the step-size or gain of frontend $i\in\mathcal F$. The latter is the continuous-time version of discrete-time projected gradient descent (see, e.g., chapter 3.5 of \citealt{aubincellina1984differential} and chapter 2 of \citealt{borkar2008stochastic}) in which routing probabilities would be updated recursively using the formula 
\begin{align}\label{eq:update-discrete}
    \bx_{i}(t+\delta t) = \Pi_{\Delta_i}\left(\bx_i(t)-\eta_i \boldsymbol{g}_i(t) \right)\,,
\end{align} 
with $\delta t > 0$ denoting the time between decision epochs. The latter update rule, which we adopt in our numerical experiments, can be used to  implement {\DLBshort} in discrete-time systems.

Some observations are in order. First, at each time step, frontend $i \in \FS$ needs to project gradients to the tangent cone of the probability simplex at $\bx_i(t)$. In Appendix~\ref{app:projection}, we give an efficient algorithm to compute the projection in $O(|\BS| \log |\BS|)$ steps. Second, the step-sizes capture how fast frontends update their routing decisions and their choice will play a key role in the stability analysis as large steps-sizes can lead to oscillatory behavior. Finally, by projecting gradients to the tangent cone of the feasible set, flow balance is always satisfied at the frontends, i.e., all requests are immediately routed to a backend.

Equations \eqref{eq:dynamics-workloads} and \eqref{eq:dynamics-delay-general} are a system of differential equations that determine the evolutions of the requests and routing probabilities under our algorithm. The following result shows this system admits a unique solution.

\begin{lem}\label{lem:existence}
There exists a unique absolutely continuous solution $(\bN(t), \bx(t))$ to the system of differential equations \eqref{eq:dynamics-workloads} and \eqref{eq:dynamics-delay-general} when the initial conditions $(\bN(t), \bx(t))$ for all $t \le0$ are absolutely continuous.
\end{lem}

Showing existence of a solution is challenging because the projection operator is discontinuous. Consider a simple network with one frontend and two backends. When the routing decisions are interior, the differential equation for $\bx_1(t)$ is non-linear but Lipschitz continuous. However, when one component of $\bx_1(t)$ becomes zero, the right-hand side of the differential equation jumps to prevent the component from going below zero. Moreover, the dynamics are complicated by the delays. All in all, dynamics are governed by a delayed differential equation with discontinuous right-hand side. 

We combine two techniques to prove the result. First, we use the method of steps to handle delays, which breaks the delayed differential equation into a sequence of differential equations without delays (see, e.g., chapter 1.2 from \citealt{hale2013introduction}). This allows us to analyze the differential equation for the workloads and routing probabilities separately. Second, we use the theory of differential inclusions to handle the discontinuity of the projection operator and exploit the fact that the projection is a maximally monotone operator (see, e.g., chapter 3 from \citealt{aubincellina1984differential}). To the best of our knowledge, this is the first existence result for this kind of dynamical system.

\subsection{Motivation}\label{sec:motivation}

We next provide a heuristic derivation of our algorithm. Suppose we fix some routing decisions $x_{ij}$ for long enough time so that backends stabilize. When flow balance is satisfied at backend $j \in \BS$, the equilibrium workloads should be equal to
\[
    N_j(\bx) := \ell_j^{-1}\left(\sum\nolimits_{i \in \FS_j} \lambda_i x_{ij}\right)\,,
\]
where the inverse exists by Assumption~\ref{assume:processing-rate}.
Recall that our objective is the average amount of requests in the system (traveling and in service), which is equal to $\sum_{j \in \BS} N_j(\bx) + \sum_{(i,j) \in \AS} \lambda_i x_{ij} \tau_{ij}$.
Using the implicit function theorem, we obtain that the partial derivative of the objective with respect to $(i,j) \in \AS$ is
\[
    \frac{\partial N_j(\bx)}{\partial x_{ij}} + \lambda_i \tau_{ij} =
    \frac {\lambda_i} {\ell_j'(N_j(\bx))} + \lambda_i \tau_{ij}\,.
\]
Under Assumption~\ref{assume:processing-rate}, the objective function is convex in $\bx$ and an optimal solution can be computed using a first-order method such as projected gradient descent. This entails each frontend $i \in \FS$ computing a ``true'' gradient of the objective with respect to its routing probabilities $\bx_i \in \mathbb R^{|\BS|}$, moving in the opposite direction of these gradient, and projecting routing decisions to the probability simplex $\Delta_i$.

In practice, we cannot implement this algorithm because ``true'' gradients of the objective are not available since we do not access to the counterfactual equilibrium workloads $N_j(\bx)$. Computing $N_j(\bx)$ requires knowing the routing decisions of all other frontends connected to a backend and global information of the processing rate function to compute the inverse $\ell_j^{-1}$. 

The main idea of {\DLBshort} is to have each frontend $i \in \FS$ approximate $N_j(\bx)$ by the most recent observed workload of the backend at time $t$, which is given by $N_j(t - \tau_{ij})$ because of network latencies. Moreover, the algorithm uses local information of the processing rate functions: it only requires the derivative at the current workloads, which can be efficiently communicated and estimated. Even though {\DLBshort} computes gradients using an approximation of the equilibrium, we shall show that the algorithm locally converges to an optimal solution in a distributed fashion.

\subsection{Optimality of Equilibrium Points}

We next prove that all equilibrium points of the dynamical system~\eqref{eq:dynamics-workloads}~and~\eqref{eq:dynamics-delay-general} are optimal for \eqref{eq:opt}. For a routing vector $\bx$ we denote the set of \emph{active arcs} by $\mathcal A(\bx) = \left\{ (i,j) \in \mathcal A : x_{ij} > 0\right\}$, which captures arcs with strictly positive flow at $\bx$.

The following lemma shows that if the projection is zero, the gradient of active arcs should be equalized and inactive arcs should have higher gradients. Because gradient descent aims to make this projection zero, it is naturally seeking points that satisfy the first-order conditions stated in Lemma~~\ref{lem:foc}. If inactive arcs had lower gradients, these would be attractive destinations for the frontend and the algorithm would route requests there instead.

\begin{lem}\label{lem:projection-is-zero} Fix a frontend $i$ and a routing probability vector $\bx_i \in \Delta_i$. Suppose that $\boldsymbol 0 = \Pi_{T_{\Delta_i}(\bx_i)}(-\eta_i \boldsymbol g_i)$. Then, there exists some $c_i \in \mathbb R$ such that $g_{ij} = c_i$ for all $(i,j) \in \mathcal A(\bx)$ and $g_{ij} \ge c_i$ otherwise.
\end{lem}

We can leverage the previous lemma to show that every equilibrium point of the algorithm is an optimal solution to the static routing problem. Therefore, our algorithm can only converge to optimal solutions.

\begin{prop}\label{prop:eq-is-opt} Every equilibrium point of {\DLBshort} is optimal for \eqref{eq:opt}.
\end{prop}

As a corollary, we obtain that at an equilibrium point, we should have that
\begin{align}
    0 &= \sum_{i\in\FS_j} \lambda_i x_{ij}^* - \ell_j (N_j^*)\,, \quad\forall j \in \BS\,, \label{eq:equilibrium-N}\\
    c_i &= \frac 1 {\ell_j'(N_j^*)} + \tau_{ij}\,,\quad\forall (i,j) \in \AS(\bx^*)\,,\label{eq:equilibrium-x}
\end{align}
for some constants $c_i > 0$ in $i \in \FS$. The second condition follows from the fact that, for each frontend, gradients are equal across active arcs.

\section{Local Stability Analysis}

We let $\|(\bN, \bx)\|_\infty = \max\left( \max_{j \in \BS} |N_j|, \max_{(i,j)\in \AS} |x_{i,j}|\right)$ be the infinity norm of a point $(\bN, \bx)$ and $(\bN^*, \bx^*)$ an optimal solution to $\operatorname{OPT}$. We present some definitions for the stability of the algorithm.

\begin{defn}
The solution of the dynamical system \eqref{eq:dynamics-delay-general} is \emph{locally stable} if for every $\epsilon > 0$ there exists some $\delta > 0$ such that if the initial conditions  satisfy $\|(\bN(t),\bx(t)) - (\bN^*, \bx^*) \|_{\infty} < \delta$ for all $t\le 0$, then for every $t\ge 0$ we have $\|(\bN(t),\bx(t)) - (\bN^*, \bx^*) \|_{\infty} < \epsilon$. Moreover, the solution is \emph{locally asymptotically stable} if it is locally stable and there exists some $\delta > 0$ such that if the initial conditions  satisfy $\|(\bN(t),\bx(t)) - (\bN^*, \bx^*) \|_{\infty} < \delta$ for all $t\le 0$, then $\lim_{t\rightarrow\infty} \|(\bN(t),\bx(t)) - (\bN^*, \bx^*) \|_{\infty} = 0$.
\end{defn}

Local stability implies that if the algorithm starts close to an optimal solution, it will remain close forever. Local asymptotic stability implies that if the algorithm starts close to an optimal solution it will remain close and converge to an optimal solution.

We analyze the stability of the algorithm under the following assumption.

\begin{assumption}\label{assume:interior-general}
There exists a unique optimal solution to the static routing problem \eqref{eq:opt} and the optimal solution satisfies strict complementary slackness, i.e.,
\[
    c_i < \frac 1 {\ell_j'(N_j^*)} + \tau_{ij}\,,\quad\forall (i,j) \in \AS \setminus \AS(\bx^*)\,.
\]
\end{assumption}

The first part of the assumption requires that the optimal solution to the static routing problem is unique. If we strengthen  Assumption~\ref{assume:processing-rate} so that processing rate functions are strictly concave, then the  optimal workloads at the backends $N_j^*$ would be unique. The routing decision $x_{ij}^*$, however, might not be unique if there exists a cycle in the network such that the sum of travel times $\tau_{ij}$ of all forward arcs is equal to the sum of the travel times of the backward arcs (forward arcs have the same direction as the cycle and backward arcs have the opposite direction). If such cycle exists, we could push flow along this cycle without changing the objective value. 
The second part of the assumption rules out a degenerate case in which gradients of inactive arcs have the same marginal impact that active ones. Slightly perturbing travel times would guarantee these conditions hold, so these are not too restrictive.

The standard approach to analyze the local stability of a non-linear dynamical system is to invoke the Principle of Linearized Stability, which involves linearizing the dynamics around an equilibrium point. This approach cannot be directly applied here because the projection operator is non-linear. We handle the projection by partitioning the local stability analysis in two steps. 

\begin{enumerate}
    \item In the first step, we show that if the arcs that are inactive at the optimal solution, which we denote by $\AS \setminus \AS(\bx^*)$, have positive flow initially, then they will have no flow in finite time. Strict complementary slackness plays a key role in this step as it guarantees that inactive arcs have negative drift.
    
    \item In the second step, we analyze the dynamics when inactive arcs have no flow and active arcs have positive flow. In this case, the projection operator can be written as a linear combination of gradients and we can use the Principle of Linearized Stability to analyze the algorithm. This is the most involved step of the proof.
\end{enumerate}

We begin with an analysis of the projection operator. By strict complementary slackness, there exists $\alpha > 0$ and a ball around the equilibrium point so that for every frontend $i \in \FS$ the gradients satisfy $\eta g_{ij} < \eta g_{ij'} - \alpha$ for all $(i,j) \in \AS(\bx^*)$ and $(i,j') \in \AS \setminus \AS(\bx^*)$. That is, the gradients of the active arcs are strictly better than those of the inactive ones. We refer to this condition as the gradients being ``well separated.'' All our analysis is conducted under the assumption that the system state is close to the equilibrium and lies in the above identified ball so that gradients are well separated. We show, moreover, that the system state never leaves the ball so that the well-separated condition holds throughout time.

The first part of Lemma~\ref{lem:projection-analysis} shows that, under this condition, arcs that are inactive at the optimal solution should have negative drift. The second part shows that, when all inactive arcs at the optimal solution have no flow, the projection admits a simple form: the inactive arcs stay at zero and the active arcs are updated by subtracting the average gradient of the active arcs. The third part shows that the projections are bounded when gradients are bounded. The proof of this result requires a careful analysis of the first-order conditions of the projection problem. In the following, we denote the set of active arcs of frontend $i$ by $\BS_i(\bx) = \left\{ j \in \BS_j : x_{ij} > 0 \right\}$.

\begin{lem}\label{lem:projection-analysis} Fix a frontend $i$ and a routing probability vector $\bx \in \Delta_i$ with $x_{ij} > 0$ for $(i,j) \in \AS(\bx^*)$. Suppose there exists some $\alpha > 0$ such that gradients satisfy $\eta g_{ij} < \eta g_{ij'} - \alpha$ for all $j \in \BS_i(\bx^*)$ and $j' \in \BS_i \setminus \BS_i(\bx^*)$. Let $\bv = \Pi_{T_{\Delta_i}(\bx_i)}(-\eta_i \boldsymbol g_i)$ be the projection of the gradient to the tangent cone.
\begin{enumerate}
    \item If $x_{ij} > 0$ for some $(i,j) \in \AS \setminus \AS(\bx^*)$, then for all $j \in \BS_i \setminus \BS_i(\bx^*)$ we have that
    \[
        \sum_{j \in \BS_i \setminus \BS_i(\bx^*)} v_j \le - \frac \alpha 2\,.
    \]

    \item If $x_{ij} = 0$ for all $(i,j) \in \AS \setminus \AS(\bx^*)$, then
    $v_j = 0$ for $j \in \BS_i \setminus \BS_i(\bx^*)$ and for $j \in \BS_i(\bx^*)$ we have
    \[
    v_j = -\eta g_{ij} + \frac 1 {|\BS_i(\bx^*)|} \sum_{j' \in \BS_i(\bx^*)}  \eta_i g_{ij'}\,.
    \]
    
    \item If $g_{ij} \in [0, \bar g]$, then $|v_j| \le \eta \bar g$ for all $j \in \BS_i$.
\end{enumerate}
\end{lem}

\subsection{Inactive Arcs}

The following result shows that if the algorithm starts close to an optimal solution, the flow of the arcs that are inactive at an optimal solution would go to zero in finite time.

\begin{lem}\label{lem:inactive-stability}
For every $\epsilon > 0$ there exists some $\delta > 0$ and $t_0>0$ such that if the initial conditions  satisfy $\|(\bN(t),\bx(t)) - (\bN^*, \bx^*) \|_{\infty} < \delta$ for all $t\le 0$, then for every $t \in [0,t_0]$ we have $\|(\bN(t),\bx(t)) - (\bN^*, \bx^*) \|_{\infty} < \epsilon$ and $x_{ij}(t_0) = 0$ for all $(i,j) \in \AS \setminus \AS(\bx^*)$.
\end{lem}

Why is establishing the local stability of inactive arcs challenging?  Consider an arc $(i,j) \in \AS \setminus \AS(\bx^*)$ that is inactive at the optimal solution (i.e., $x_{ij}^* = 0$) and has positive flow initially (i.e., $x_{ij}(0) > 0$). Because inactive arcs have large gradients at the optimal solution and the initial state is close to the optimal solution, this arc must have a large gradient and the algorithm should reduce its flow. Unfortunately, this might not happen initially. If there is another arc $(i,j') \in \AS \setminus \AS(\bx^*)$ with even larger gradient, the projection operator would prioritize decreasing the flow of $(i,j')$ first. In fact, the flow of $(i,j)$ might even increase initially! Eventually, when the flow of $(i,j')$ goes to zero at some later time $t>0$, arc $(i,j)$ could become the arc with the largest gradient and only then its flow would be guaranteed to decrease. This delayed reduction in flow explains the difficulty in establishing the stability of inactive arcs.

We prove the result by considering the Lyapunov function 
\[
    V(\bx) = \sum_{(i,j) \in \AS \setminus \AS(\bx^*)} x_{ij}\,,
\]
which gives the flow of all arcs that are inactive at the optimal solution. In light of Lemma~\ref{lem:projection-analysis}, part~1, we can show that the total drift of the arcs that are inactive at the optimal solution is negative and the Lyapunov function has a negative drift, i.e., $\frac d {dt} V(\bx(t)) < 0 $. This property follows because, for each frontend, gradients are well separated, i.e., the gradients of arcs that are inactive at the optimal solution are strictly larger than those of the active arcs. Well separateness, in turn, follows because processing rate functions are smooth and strict complementary slackness holds. 

Because our Lyapunov function has negative drift,  the flows of the arcs in $\AS \setminus \AS(\bx^*)$ go to zero in finite time. By looking at the total flow across all arcs in $\AS \setminus \AS(\bx^*)$, we can avoid dealing with the complicated dynamics of individual arcs, which might be non-monotone as previously discussed. We also show that flows of the arcs in $\AS(\bx^*)$ and workloads do not change too much in the time it takes the other arcs' flows to go to zero using some basic estimates on their growth.

\subsection{Active Arcs}

We next move to the most challenging part of our analysis---studying the local stability of {\DLBshort} across active arcs. By Lemma~\ref{lem:inactive-stability}, the flow of the arcs that are inactive at the optimal solution would go to zero in a finite amount of time $t_0 > 0$. 
Part 2 of Lemma~\ref{lem:projection-analysis} guarantee that arcs that are inactive at the optimal solution would remain at zero as long as gradients remain well separated. Therefore, we now restrict attention to the case when all inactive arcs have no flow and active arcs have positive flow. Moreover, we assume that the network is connected. Otherwise, each connected component can be analyzed independently. 

Before stating our main result, we state some definitions. Let $A \in \{0,1\}^{|\FS|\times|\BS|}$ the bi-adjacency matrix of the network restricted to arcs $\AS(\bx^*)$ that are active at an optimal solution, i.e., $a_{i,j} = 1$ if $(i,j) \in \AS(\bx^*)$ and zero otherwise. Let $\ba_i^\top$ be the $i$-th row of the bi-adjacency matrix $A$, which captures the connectivity of a frontend $i \in \FS$. The matrix $E_i \in \mathbb R^{|\BS|\times|\BS|}$ given by
\begin{align}\label{eq:laplacian-matrix}
        E_i = \diag(\ba_i) - \frac{\ba_i \ba_i^\top}{\ba_i^\top \mathbf 1}\,.
\end{align}
is the Laplacian matrix of the subgraph induced by frontend $i \in \FS$. We denote by \emph{gap} the spectral gap or minimum non-zero eigenvalue of the weighted Laplacian matrix $\sum_{i\in\FS}\lambda_i \eta_i E_i$. The following theorem states our main  result.

\begin{thm}\label{thm:stability-general}
{\DLBshort} is locally asymptotically stable when the network $\mathcal G$ is connected if there exists some $\hat c \in \mathbb R_+$ with $\hat c \ge 1/\ell_j'(N_j^*)$ for all $j \in \BS$ such that
\begin{align}\label{eq:stability-general}
    2 \boldsymbol{\eta}^\top \boldsymbol{\lambda} \left(
    \max_{j \in \BS} \frac{\hat \tau_j  \sigma_j}{\ell_j'}
    + \frac{\sum_{i \in \FS} \lambda_i \eta_i |\hat c - c_i|}{\operatorname{gap}} \cdot \hat c \cdot \max_{j \in \BS} \sigma_j \right)  < 1\,,
\end{align}
where $\ell_j' = \ell_j'(N_j^*)$ is the derivative of the processing rate function at the optimal workload, $\sigma_j = - \ell_{j}''(N_{j}^*) / \ell_{j}'(N_{j}^*)^2 > 0$ captures the curvature of the processing rate function at the optimal workload, and $\hat \tau_j = \hat c - 1/ \ell_j'(N_j^*)$ is the \emph{uniform latency} of backend $j$.
\end{thm}

Before providing the proof of the main result, which we postpone to Section~\ref{sec:proof-main}, we discuss the implications of Theorem~\ref{thm:stability-general}. First, in the case of single frontend, which we denote by $i$, we can set $\hat c = c_i$. This leads to $\hat \tau_j = \tau_{ij}$ and condition~\eqref{eq:stability-general} reduces to
\begin{align}\label{eq:stability-condition1}
    \max_{j \in \BS} \frac{2 \tau_{ij} \eta_i \lambda_i \sigma_j}{\ell_j'} < 1\,.
\end{align}
The main takeaway is that we need to pick the step-size to be inversely proportional to the delay to guarantee stability. This is intuitive: if a backend has a long delay, we should pick a smaller step-size to avoid being too reactive. The dependence on the arrival rate is more subtle because there are other parameters in the condition (such as $\ell_j'$ and $\sigma_j$) involving the first and second derivative of the processing rate function evaluated at the optimal workloads, which change with the arrival rate.

In the case of multiple frontends, our analysis relies on a novel uniformization technique that combines the different frontends with Lagrange  multipliers $c_i$ into a single representative frontend with a common Lagrange multiplier $\hat c$, which we refer to as the ``pivot.'' For this new frontend to satisfy the first-order condition \eqref{eq:equilibrium-x} we set the latencies of this virtual frontend to be the uniform latencies $\hat \tau_j = \hat c - 1/ \ell_j'(N_j^*)$ in the statement of the theorem.

Suppose the network $\mathcal G$ is complete, but only some arcs have flow in the optimal solution. A natural, albeit potentially suboptimal, choice is to set $\hat c = \max_{i \in \FS} c_i$. This choice is clearly feasible and gives that the uniform latency of backend $j \in \BS$ is at most the largest latency across all frontends connected to it. More formally,
\[
    \hat \tau_j = \hat c - \frac 1 {\ell_j'} = \max_{i \in \FS} \left\{ c_i - \frac 1 {\ell_j'}\right\} \le \max_{i \in \FS} \tau_{ij}\,,
\]
where the inequality follows from the first-order optimality conditions in Lemma~\ref{lem:foc}. Now, a sufficient condition for stability is
\[
    \max_{i\in \FS,j \in \BS} \frac{2 \tau_{ij} \boldsymbol{\eta}^\top \boldsymbol{\lambda} \sigma_j}{\ell_j'}
    + \frac{2 (\boldsymbol{\eta}^\top \boldsymbol{\lambda})^2 \max_{i \in \FS} c_i^2}{\operatorname{gap}} \cdot \max_{j \in \BS} \sigma_j  < 1\,,
\]
The first term is easy to interpret as it naturally extends the single-frontend stability condition presented in \eqref{eq:stability-condition1}. The second term is harder to interpret as it depends on the Lagrange multipliers $c_i$ of the frontends' flow balance constraints. From the first-order optimality conditions, we know that $c_i = \tau_{ij} + 1/\ell_j'$, so larger travel times would decrease the range of stability.

We conclude this section by providing a lower bound on the spectral gap of the Laplacian using techniques from spectral graph theory. Let $\mathcal P$ be a path between two backends restricted to arcs in $\AS(\bx^*)$. We denote by $\BS(\mathcal P)$ and $\FS(\mathcal P)$ to be the set of backends and frontends, respectively, visited by the path $\mathcal P$. We define the \emph{length} of the path to be $d(\mathcal P) = \sum_{i \in \FS(\mathcal P)} |\BS_i| / (\lambda_i \eta_i)$. The \emph{diameter} of the network is the length of the longest shortest path between backends, i.e.,
\[
    d(\mathcal G) = \max_{j,j' \in \BS} \min_{\substack{\mathcal P \text{ path}\\\text{from } j \text{ to } j'}} d(P)\,.
\]

\begin{lem}\label{lem:spectral-bound}Suppose the network $G$ is connected. Then, the spectral gap or minimum non-zero eigenvalue of the network satisfies
    \[
    \operatorname{gap}\left(\sum_{i\in\FS}\lambda_i \eta_i E_i\right) \ge \frac 1{|\BS| \cdot d(\mathcal G)}\,.
\]
\end{lem}

\section{Proof of Main Result}\label{sec:proof-main}

By Lemma~\ref{lem:inactive-stability}, the flow of the inactive arcs in $\AS \setminus \AS(\bx^*)$ goes to zero in a finite amount of time $t_0 > 0$ and active arcs in $\AS(\bx^*)$ remain active at $t_0$. We assume without loss of generality that $t_0 = 0$. Therefore, we have that $x_{ij}(0) = 0$ for all $(i,j) \not \in \AS(\bx^*)$ and $x_{ij}(t) > 0$ for all $(i,j) \in \AS(\bx^*)$ and $t \in [-\tau_{ij},0]$.

If the system is stable, the workloads will stay in a ball around $\bN^*$ and, by Assumptions~\ref{assume:processing-rate}~and~\ref{assume:interior-general}, gradients remain well separated. Part 2 of Lemma~\ref{lem:projection-analysis} guarantee that the drift of inactive arcs in $\AS \setminus \AS(\bx^*)$ would remain at zero as long as gradients remain well separated. Therefore, if the workloads remain in the ball around the equilibrium point, we would have that $x_{ij}(t) = 0$ for all $(i,j) \not \in \AS(\bx^*)$ and $t > 0$. To simplify the exposition we assume without loss that all arcs in the network are active as inactive arcs remain at zero and can be simply ignored. We set $\AS(\bx^*) = \AS$ and $\BS_i(\bx^*) = \BS_i$.\footnote{Lemma~\ref{lem:inactive-stability} only guarantees that $x_{ij}(t) = 0$ for all $(i,j) \not \in \AS(\bx^*)$ at $t=t_0$ but not necessarily at $t < t_0$. Because of delays introduced by travel times, a backend $j\in\BS$ can receive requests from an arc $(i,j) \in \AS \setminus \AS(\bx^*)$ during $t \in [t_0, t_0 + \tau_{ij}]$. The local stability of a delay differential equation depends on the functional operator and not on the initial conditions as long these are close to the equilibrium point. Thus, we can safely ignore the lagged flows from arcs $\AS \setminus \AS(\bx^*)$.}

Using part 2 of Lemma~\ref{lem:projection-analysis} we can write the dynamics as follows
\begin{equation}\label{eq:dynamics-delay-active}
\begin{split}
    \frac {d} {dt} N_j(t) &= \sum_{i\in\FS_j} \lambda_i x_{ij}(t - \tau_{ij}) - \ell_j (N_j(t))\,, \quad \forall j \in \BS\,, \\
    \frac {d} {dt} x_{ij}(t) &= -\eta_i g_{ij}(t) + \frac 1 {|\BS_i|} \sum_{j' \in \BS_i} \eta_i g_{ij'}(t)\,, \quad \forall (i,j) \in \AS\,. \\
\end{split}
\end{equation}
In our analysis, we can now safely ignore the projection to the probability simplex and assume that routing decisions are interior. That is, $x_{ij}(t) > 0$ and $\sum_j x_{ij}(t) = 1$.

\subsection{Principle of Linearized Stability}
We study the local stability of the dynamical system~\eqref{eq:dynamics-delay-active} around the equilibrium point $(\bN^*,\bx^*)$ by looking at its linearization.

Let $\bar N_j(t) = N_j(t) - N_j^*$ and $\bar x_{ij}(t) = x_{ij}(t) - x_{ij}^*$ be the deviations around the equilibrium point. Performing a first-order expansion of \eqref{eq:dynamics-delay-active} around the equilibrium point we obtain the following linear dynamics for the workloads
\begin{equation}\label{eq:linear-dynamics-N}
\begin{split}
    \frac {d} {dt} \bar N_j(t) &=
    \frac {d} {dt} N_j(t)
    \approx
    \sum_{i\in\FS_j} \lambda_i x_{ij}^* - \ell_j (N_j^*)
    + \sum_{i\in\FS_j} \lambda_i \bar x_{ij}(t - \tau_{ij}) - \ell_j' (N_j^*) \bar N_j(t)\\
    &= \sum_{i\in\FS_j} \lambda_i \bar x_{ij}(t - \tau_{ij}) - \ell_j' \bar N_j(t)\,,
\end{split}
\end{equation}
where the last equation follows from the equilibrium condition~\eqref{eq:equilibrium-N} and using $\ell_j'>0$ as shorthand for $\ell_j' (N_j^*)$. Similarly, for the routing probabilities we obtain that
\begin{equation}\label{eq:linear-dynamics-x}
\begin{split}
    \frac {d} {dt} \bar x_{ij}(t) &=
    \frac {d} {dt} x_{ij}(t) \approx
    - \eta_i \left( \frac 1 {\ell_j'(N_j^*)} + \tau_{ij} \right) + \frac {\eta_i} {|\BS_i|}
    \sum_{j' \in \BS_i} \left( \frac 1 {\ell_{j'}'(N_{j'}^*)} + \tau_{ij'}\right)\\
    &\quad + \eta_i \frac {\ell_j''(N_j^*)}{\ell_j'(N_j^*)^2} \bar N_j(t - \tau_{ij}) - \frac {\eta_i} {|\BS_i|}
    \sum_{j' \in \BS_i} \frac {\ell_{j'}''(N_{j'}^*)}{\ell_{j'}'(N_{j'}^*)^2} \bar N_{j'}(t - \tau_{ij'})\\
    &= - \eta_i \sigma_j \bar N_j(t - \tau_{ij}) + \frac {\eta_i} {|\BS_i|}
    \sum_{j' \in \BS_i} \sigma_{j'} \bar N_{j'}(t - \tau_{ij'})\,,
\end{split}
\end{equation}
where the last equation follows from setting $\sigma_j = - \ell_{j}''(N_{j}^*) / \ell_{j}'(N_{j}^*)^2 > 0$ and the equilibrium condition~\eqref{eq:equilibrium-x}.

For a real function $f$, we define its Laplace transform by $\mathcal L\{f\}(s) = \int_0^\infty e^{-st}f(t) dt$ where $s \in \mathbb C$ is a complex number. For a complex number $s \in \mathbb C$, we denote by $\Re(s)$ its real part and by $\Im(s)$ its imaginary part.

Let $\LN_j(s) = \mathcal L\{\bar N_j\}(s)$ and $\Lx_{ij}(s) = \mathcal L\{\bar x_{ij}\}(s)$ be the Laplace transforms of the deviations of workloads and routing probabilities, respectively, around their equilibrium points. Taking Laplace transforms to the linearized dynamics, we obtain that
\begin{align}
    s \LN_j(s) - \bar N_j(0) &= \sum_{i\in\FS_j} \lambda_i \Lx_{ij}(s) e^{-\tau_{ij}s} - \ell_j' \LN_j(s)\label{eq:laplace-N}\\
    s \Lx_{ij}(s) - \bar x_{ij}(0) &=
    - \eta_i \sigma_j \LN_j(s) e^{- \tau_{ij}s} + \frac {\eta_i} {|\BS_i|}
    \sum_{j' \in \BS_i} \sigma_{j'} \LN_{j'}(s) e^{-s \tau_{ij'}}\,,\label{eq:laplace-x}
\end{align}
where we used the derivative, time-shifting, and linearity properties of the Laplace transform. We solve for $\Lx_{ij}$ by multiplying \eqref{eq:laplace-N} by $s$ and then using \eqref{eq:laplace-x} to  obtain the following equations for each backend
\begin{align*}
    &s^2 \LN_j(s) + s \ell_j' \LN_j(s) +  \sum_{i\in\FS_j} \lambda_i \eta_i \sigma_j \LN_j(s)  e^{-2\tau_{ij}s}
    -  \sum_{i\in\FS_j} \frac{\lambda_i \eta_i e^{-s \tau_{ij} }}{|\BS_i|} \sum_{j' \in \BS_i} \sigma_{j'} \LN_{j'}(s) e^{-s\tau_{ij'}}\\
    &= s \bar N_j(0) +
    \sum_{i\in\FS_j} \lambda_i \bar x_{ij}(0) e^{-\tau_{ij}s}\,.
\end{align*}
Let $\boldsymbol{\LN}(s) = (\LN_j(s))_{j \in \BS}$ be the vector of workloads Laplace transforms and $H(s) = (H_j(s))_{j \in \BS}$ be right-hand vector given by $H_j(s) = s \bar N_j(0) + \sum_{i\in\FS_j} \lambda_i \bar x_{ij}(0) e^{-\tau_{ij}s}$. The previous equation can be written in matrix form as follows
\[
\left( s^2 I + s \diag(\boldsymbol{\ell'}) + \sum_{i \in \mathcal F} \lambda_i \eta_i Q_i(s) \diag(\boldsymbol{\sigma}) \right) \boldsymbol{\LN}(s) = H(s)\,,
\]
where $I \in \mathbb R^{|\BS|\times|\BS|}$ is the identity matrix, $\boldsymbol{\ell'} \in \mathbb R^{|\BS|}$ is a vector with $\ell_j = \ell_j'(N_j^*)$ in the $j$-th component, $\boldsymbol{\sigma} \in \mathbb R^{|\BS|}$ is a vector with $\sigma_j$ in the $j$-th component, and
\[
Q_i(s) = \diag(\br_i(s)) E_i \diag(\br_i(s))
\]
where the Laplacian matrix $E_i$ is defined in \eqref{eq:laplacian-matrix} and $\br_i(s)^\top$ is the $i$-th row of the delay matrix $R : \mathbb C \rightarrow \mathbb C^{|\FS|\times|\BS|}$ given by $r_{ij}(s) = \exp(-s \tau_{ij})$ if $(i,j) \in \AS$ and zero otherwise. In the previous equation, we used that $|\BS_i| = \sum_{j \in \BS_i} 1 = \ba_i^\top \mathbf 1$.

To guarantee local asymptotic stability,
we need to check that all solutions $s \in \mathbb C$ to the characteristic equation
\[
\det \left( s^2 I + s \diag(\boldsymbol{\ell'}) + \sum_{i \in \mathcal F} \lambda_i \eta_i Q_i(s) \diag(\boldsymbol{\sigma}) \right) = 0
\]
have negative real parts (see, e.g., \citealt[Theorem 4.8, p.55]{smith2011introduction}).{.

A challenge is that $s=0$ can be a solution to the characteristic equation because
\[
\det \left( \sum_{i \in \mathcal F} \lambda_i \eta_i Q_i(0) \diag(\boldsymbol{\sigma}) \right) = \det \left(\sum_{i \in \mathcal F} \lambda_i \eta_i Q_i(0) \right) = 0
\]
since $\mathbf 1$ lies in the null space of $Q_i(0)$ (see Lemma~\ref{lem:laplacian-matrix2}). Using the fact that $\sum_{j \in \BS_i} x_{ij}(t) = 1$, we can show that $s = 0$ cannot lead to a valid solution to the differential equation. Otherwise, we would contradict the uniqueness of the optimal solution of the static routing problem from Assumption~\ref{assume:interior-general}.

\begin{lem}\label{lem:zero-solution} The only solution associated to $s = 0$ is $\bar x_{ij}(t) = 0$ and $\bar N_j(t) = 0$.
\end{lem}

In addition, because solutions $s = - \ell_j' < 0$ have negative real parts, we can factor out $s^2 I + s \diag(\bell')$ of the characteristic equation and check if the roots of the equation $\det(I + L(s)) = 0$ have negative real parts where the ``loop transfer'' function $L(s)$ is given by
\begin{align}\label{eq:loop-function}
    L(s) =  \left(s^2 I +  s \diag(\bell') \right)^{-1} \sum_{i \in \mathcal F} \lambda_i \eta_i Q_i(s) \diag(\boldsymbol{\sigma})\,.
\end{align}

\subsection{Generalized Nyquist Criterion}

Using \eqref{eq:equilibrium-x}, we obtain that travel times are additively separable and decompose the delay matrix as follows
\[
    r_{ij}(s) = \exp(-s \tau_{ij} ) = \exp(-s c_i) \exp(s / \ell_j')\,,
\]
where $c_i$ is the Lagrange multiplier of the flow balance constraint of frontend $i \in \FS$. Defining $\balpha : \mathbb C \rightarrow \mathbb C^{|\BS|}$ to be the vector function $\alpha_j(s) = \exp(s/\ell_j')$, we can write
\[
    Q_i(s) = \exp(-2 s c_i) \diag(\balpha(s)) E_i \diag(\balpha(s))\,.
\]

We apply the latter decomposition to the loop function~\eqref{eq:loop-function}. Because $\det(I + AB) = \det(I + BA)$ for square matrices $A,B$ together with the fact that diagonal matrices are multiplicatively commutative, we can look at the characteristic equation $\det(I + \hat L(s)) = 0$ with
\begin{align}\label{eq:hat-loop}  
    \hat L(s) = \left( \sum_{i \in \mathcal F}
    \exp(-2 s c_i) \lambda_i \eta_i E_i \right) D(s)
\end{align}
where the diagonal matrix $D(s)$ is given by
\[
    D(s) = \left( s^2 I + s \diag(\boldsymbol{\ell'})\right)^{-1}  \diag(\boldsymbol \sigma) \diag(\balpha(2s))\,.
\]

From the Generalized Nyquist criterion~\citep{desoer1980generalized}, it is sufficient to check that all eigenloci (the loci traced by eigenvalues as we vary the frequency) of the loop transfer function $\hat L(s)$ cross the real line to the right of the point $-1 + \im 0$ for $s = \im w$ with $w \in \mathbb R$ and $\im$ the imaginary unit. 

The standard approach to analyze the eigenvalues of the loop transfer function $\hat L(s)$ is to decompose it as a product of a positive semi-definite matrix and a diagonal matrix~\citep{srikant2004mathematics,low2022analytical}. Because the constants $c_i$ are different across frontends, the matrix $\sum_{i \in \mathcal F}
\exp(-2sc_i) \lambda_i \eta_i E_i$ in the loop transfer function $\hat L(s)$ is not necessarily positive semi-definite and we cannot apply standard results from spectral analysis. As a warm-up, we first consider the single-frontend case, in which we can simply decompose the loop transfer function as a product of a positive semi-definite matrix and a diagonal matrix because the aforementioned sum has a single term. In the general case, we analyze the eigenvalues of the loop function using a novel uniformization technique that decomposes the matrix into a positive semi-definite matrix and a “small” perturbation. Because of space considerations, we postpone the  general case to Appendix~\ref{sec:general-case}.

\subsection{Single-Frontend Case}

To develop some intuition, we first consider the case when there is only one frontend. We begin by presenting some fundamental facts about eigenvalues and the numerical range of a matrix. These are borrowed from chapter 1 of \citet{horn1991topics}. For a square complex matrix $A \in \mathbb C^{n \times n}$, we denote by  $\operatorname{spec}(A)$ the spectrum or set of eigenvalues. Let $A^\dagger = \overline{A^\top}$ denote the Hermitian conjugate and $W(A) = \left\{ \bx^\dagger A \bx : \bx \in \mathbb C^n, \bx^\dagger \bx = 1\right\}$ denote the numerical range or field of values of $A$. The numerical range is a subset of complex numbers, i.e., $W(A) \subseteq \mathbb C$. We state some properties of the spectrum and numerical range.
\begin{itemize}
    \item The numerical range contains all eigenvalues, i.e., $\operatorname{spec}(A) \subseteq W(A)$.
    \item The numerical range is subadditive, i.e., $W(A + B) \subseteq W(A) + W(B)$ where the sum on the right-hand side is the Minkowski sum of two sets.
    \item If $A$ is positive semi-definite then $\operatorname{spec}(A B) \subseteq W(A) W(B)$. Moreover, the numerical range of $A$ is $W(A) = [\min \operatorname{spec}(A), \max \operatorname{spec}(A)]$ with $\operatorname{spec}(A) \subseteq \mathbb R_+$.
    \item The numerical range of a diagonal matrix is the convex hull of its entries, i.e., $W(\diag(\boldsymbol a)) = \operatorname{conv}\{ a_i : i = 1,\ldots,n\}$.
    \item The numerical range of $A$ is contained in a disk of radius $\|A\|$, i.e., $W(A) \subseteq \operatorname{disk}(\|A\|)$ where $\operatorname{disk}(r) = \{ z \in \mathbb C : |z| \le r\}$ and  $\|\cdot\|$ is the spectral norm of a matrix.
\end{itemize}

We begin by showing some fundamental facts about the Laplacian matrix $E_i$.

\begin{lem}\label{lem:laplacian-matrix1}The Laplacian matrix $E_i = \diag(\ba_i) - (\ba_i \ba_i^\top) / (\ba_i^\top \mathbf 1)$ is positive semi-definite with spectral radius at most one.
\end{lem}

Armed with the previous facts, we can analyze the spectrum of $\hat L(s)$ as follows:
\begin{equation}\label{eq:spec-bound}
\begin{split}
    \operatorname{spec}(\hat L(s))
    &= \operatorname{spec} \left( 
    \exp(-2 s c_i) \lambda_i \eta_i E_i  D(s) \right) \\
    &\subseteq 
    W(\exp(-2 s c_i) \lambda_i \eta_i   D(s)) \\
    &=  \operatorname{conv}\Bigg(0, \underbrace{\frac{ \lambda_i \eta_i \sigma_j e^{-2\tau_{ij} s}}{s^2 + s\ell_j'}}_{\hat L_{ij}(s)} : (i,j) \in \AS \Bigg)\,,
\end{split}
\end{equation}
where the first inclusion follows from the bound on the spectrum of the product of a positive semi-definite matrix and an arbitrary matrix together with  Lemma~\ref{lem:laplacian-matrix1}, and the second equality because the numerical range of a diagonal matrix is the convex hull of its entries and using the formula for the diagonal matrix $D(s)$ together with $\tau_i = c_i - 1/\ell_j'$ from equation~\eqref{eq:equilibrium-x}.

Fix an arc $(i,j) \in \AS$ and let $s = \im w$. Then,
\begin{align*}
     \hat L_{ij}(\im w) &
    = \frac{\lambda_i \eta_i\sigma_j e^{-2\tau_{ij} \im w}}{\im w(\im w + \ell_j')}
     = \frac{2\tau_{ij} \lambda_i \eta_i\sigma_j}{ \ell_j'} \frac{2\tau_{ij} \ell_j' e^{-2\tau_{ij} \im w}}{2\tau_{ij} \im w(2\tau_{ij}\im w + 2\tau_{ij}\ell_j')}\\
     &= \frac{2\tau_{ij} \lambda_i \eta_i\sigma_j}{\ell_j'} \frac{\theta_{ij} e^{- \im x}}{ \im x(\im x + \theta_{ij})}\,,
\end{align*}
where the last equation follows form setting $\theta_{ij} = 2\tau_{ij} \ell_j'$ and $x = 2\tau_{ij} w$. If we assume that condition \eqref{eq:stability-condition1} holds, 
we obtain by Lemma 5.6 of \citet{srikant2004mathematics} that $\hat L_{ij}(\im w)$ crosses the real line at a point to the right of $-1 + \im 0$. See Figure~\ref{fig:nyquist} for a Nyquist plot of the function. Unfortunately, this condition does not immediately imply the result because the convex hull operation might result in crossings to the left of $-1 + \im 0$. This condition is clearly sufficient when all the $\hat L_{ij}(s)$ are the same, i.e., in the symmetric case when $\tau_{ij}$ and $\ell_j$ are equal across arcs and backends, respectively. Interestingly, this fact is overlooked in standard treatments on the stability analysis of TCP congestion control algorithms. Some papers such as \citet{paganini2005congestion,tian2006stability} do perform a more careful analysis of the convex hulls of eigenvalues, but their analysis does not apply directly to our setting.

In the case of single frontend with asymmetric backends, we need a more refined analysis that leverages the first-order optimality conditions of the static routing problem in \eqref{eq:equilibrium-x}.

\begin{figure}
    \centering
    \begin{subfigure}[t]{0.45\textwidth}
    \centering
    \scalebox{0.8}{ 
    \begin{tikzpicture}[every node/.style={scale=1}] 
        \newcommand{\thetaval}{1}
        \begin{axis}[axis lines = middle, xmin = -2, xmax = 0.5, ymin=-1, ymax=1]
        \addplot[domain=0.5:20,samples=500,thick]({-\thetaval*(cos(deg(x))*x + sin(deg(x))*\thetaval)/x/(x^2 + \thetaval^2)},{-\thetaval*(cos(deg(x))*\thetaval - sin(deg(x))*x)/x/(x^2 + \thetaval^2)});
        \node[fill, circle,inner sep=1.5pt] at (axis cs: -1,0) {};
        \end{axis}
    \end{tikzpicture}
    }
    \caption{Plot of $w \mapsto \theta e^{-\im w} / (\im w (\im w + \theta))$ for $\theta = 1$ as $w$ increases from 0 to $\infty$. The curve crosses the real line to the right of $-1 + \im 0$, which is marked with a dot. }
    \label{fig:nyquist}
    \end{subfigure}%
    \hspace{2em}
    \begin{subfigure}[t]{0.45\textwidth}
    \centering
    \scalebox{0.8}{
    \begin{tikzpicture}[every node/.style={scale=1}]
        \newcommand{\wval}{1}
        \begin{axis}[axis lines = middle, legend pos = south east, legend entries = {$c_1=2$\\$c_2=4$\\}, 
        xmin = -2, xmax = 0.5, ymin=-1, ymax=1]
        \foreach \cval in {2,4}{
        \addplot+[mark=none,domain=0.01:\cval,samples=100,thick]({-(cos(deg(2*x*\wval))*\wval*(\cval-x) + sin(deg(2*x*\wval))) / (2*\wval*x*(\cval^2*\wval^2 - 2*\cval*x*\wval^2+x^2*\wval^2+1))},
        {-(-sin(deg(2*x*\wval))*\wval*(\cval-x) + cos(deg(2*x*\wval))) / (2*\wval*x*(\cval^2*\wval^2 - 2*\cval*x*\wval^2+x^2*\wval^2+1))});
        }
        \addplot[thick,mark=none,dashed,blue, name path = first]
        coordinates{(-2,-0.5) (1, 1)};
        \addplot[thick,mark=none,dashed,red, name path = second]
        coordinates{(-2,-0.25) (1, 0.5)};
        \path [name path=bottom]
        (axis cs:-2,-1) -- (axis cs: 0.5,-1);
        \addplot+[blue!30,fill opacity=0.2] fill between[of=first and bottom];
        \addplot+[red!30, fill opacity=0.2] fill between[of=second and bottom];
        \end{axis}
    \end{tikzpicture}
    }
    \caption{Plot of $\tau \mapsto e^{-2\tau \im w} / (2\tau \im w (\im w (c- \tau) + 1))$ for $\tau \in [0, c)$ with $w=1$, $c \in \{2,4\}$. The separating hyperplanes are shown as a dashed line and the half-spaces as shaded areas.}
    \label{fig:convexhull}
    \end{subfigure}
    \caption{Nyquist plots. The horizontal axis is the real part and the vertical axis is the imaginary part.}
\end{figure}
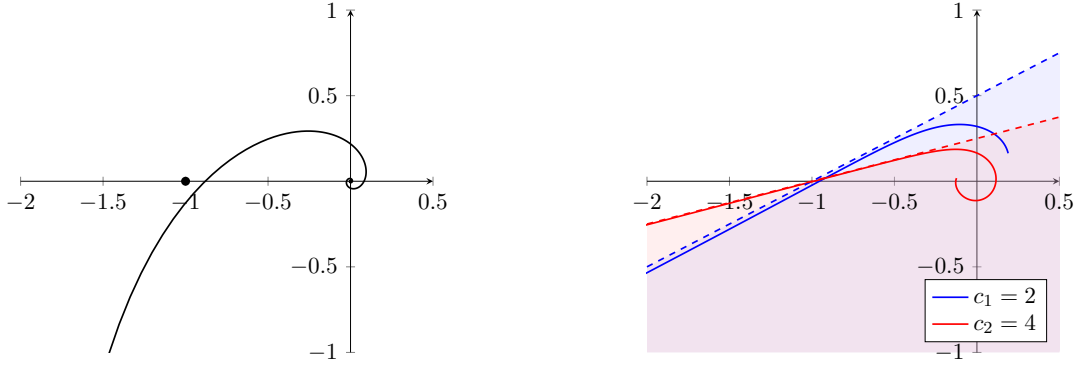

From the first-order conditions, we know that $\tau_{ij} < c_i$ and $\ell_j' = 1/(c_i - \tau_{ij})$. Therefore, we can write
\[
 \hat L_{ij}(\im w)
    = \frac{\lambda_i \eta_i\sigma_j}{\ell_j'} \frac{e^{-2\tau_{ij} \im w}}{\im w(\im w/\ell_j' + 1)}
     = \frac{2\tau_{ij} \lambda_i \eta_i\sigma_j} {\ell_j'} \frac{e^{-2\tau_{ij} \im w}}{2\tau_{ij} \im w (\im w (c_i - \tau_{ij}) + 1)}\,.
\]
Because $\hat L_{ij}(\im w)$ is conjugate symmetric, i.e., $\hat L_{ij}(\im w) = \overline{\hat L_{ij}(-\im w)}$, we only need to check results for $w \in [0, \infty)$ as results for negative values follow by symmetry.

Our analysis will study the convex hulls of these curves for a fixed value of $w$ and all possible values for the latencies $\tau$ simultaneously. The following lemma provides a geometric analysis of this curve for a fixed frontend and shows the loci generated by varying $\tau$ lies below a halfspace that goes through $-1 + \im 0$. Figure~\ref{fig:convexhull} illustrates the lemma. While the geometric fact stated in the lemma is simple, the proof is complex and requires careful trigonometric calculus. 

\begin{lem}\label{lem:geometric-bounds}Fix $c>0$ and $w \ge 0$. Consider the complex function $f : [0,c) \rightarrow \mathbb C$ given by
\[
    f(\tau) = \frac{e^{-2\tau \im w}}{2\tau \im w (\im w (c- \tau) + 1)}\,.
\]
Then, the function $f(\tau)$ lies below the line with slope $1/(wc)$ that goes through $-1 + \im 0$,  i.e., $\Re(f(\tau)) \ge -1 + wc \cdot \Im(f(\tau))$.
\end{lem}

Therefore, for each $w>0$ there exists a hyperplane separating the convex hull from the point $-1 + \im 0$. We remark that the separating hyperplanes tilt as we change the value of $w$, so we need to construct a separating hyperplane for every value of $w$. More formally, Lemma~\ref{lem:geometric-bounds} implies that under condition~\eqref{eq:stability-condition1} for each frontend $i \in \mathcal F$ we have
\[
\operatorname{conv}\left(0, \hat L_{ij}(\im w) : j \in \BS_i \right)
\subseteq \left\{ z \in \mathbb C: \Re(z) \ge -1 + wc_i \cdot \Im(z) \right\}\,,
\]
and, thus, the corresponding convex hulls in \eqref{eq:spec-bound} cannot encircle $-1 + \im 0$. Therefore, condition~\eqref{eq:stability-condition1} is sufficient in the case of a single frontend. We have proved the following result.

\begin{prop}\label{prop:stability-one-frontend}
{\DLBshort} is locally asymptotically stable in the case of a single frontend if condition~\eqref{eq:stability-condition1} holds.
\end{prop}

\section{Numerical Results}

In this section, we numerically explore the stability and optimality of our algorithm under different settings using simulation. In our stability analysis, we explore whether the system variables (routing probabilities and workloads) reach an equilibrium. In our optimality analysis, we study whether the algorithm converges to an optimal point and its performance, as measured by the workloads at the backends, along the path. We also compare {\DLBshort} with other load balancing algorithms.

We conduct three sets of experiments. In the first set of experiments, we study a simple network with one frontend and two backends, where we visualize the algorithm's stability by looking at the evolution of workloads. In the second set of experiments, we programmatically explore the local stability and optimality under different system topologies. In the third set of experiments, we explore the global stability {\DLBshort} and compare its performance with other algorithms.

To simulate the algorithm, we used Euler method with linear interpolation to handle delays and projected gradient descent as described in \eqref{eq:update-discrete} to handle discrete time steps. The projection operator is implemented using the sort algorithm from \citet{blondel2014large}. 

\subsection{Stability in a One-Frontend Two-Backend Network}

We consider the one-frontend-two-backend network depicted in Figure~\ref{fig:simple-model}. To simplify the analysis we assume processing rate functions have a squared-root dependence on workloads: $\ell_j(N_j) = \sqrt{a_j +  b_j N_j } - \sqrt{a_j}$ with $a_j > 0$ and $b_j>0$. These processing functions are increasing, concave, and twice differentiable. Moreover, they satisfy that $-\ell_j''(N_j)/\ell_j'(N_j)^3 = 2/b_j$, which is independent of the workload. We illustrate them in Figure~\ref{fig:processing-rate-function1}.

\begin{figure}[t]
    \centering
\begin{tikzpicture}
    \node[circle, draw] (n1) at (0,-1) {$f_1$};
    \node[circle, draw] (l1) at (2,0) {$b_1$};
    \node[circle, draw] (l2) at (2,-2) {$b_2$};
    \node at (4,0) {$\ell_1(N_1)$};
    \node at (4,-2) {$\ell_2(N_2)$};
    \draw[->] (n1) -- (l1) node [midway,above,sloped] {$\tau_{11}$};
    \draw[->] (n1) -- (l2) node [midway,above,sloped] {$\tau_{12}$};
    \node (lambda1) at (-1.5,-1) {$\lambda_1$};
    \draw[->] (lambda1) -- (n1);
\end{tikzpicture}
    \caption{A network with one frontend and two backends.}
    \label{fig:simple-model}
\end{figure}
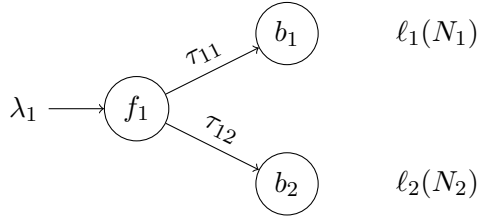

In this case, the stability condition \eqref{eq:stability-condition1} reduces to
\[
    \max_{j \in \BS} \frac{\tau_{1j} \eta_1 \lambda_1}{b_j} < 1\,,
\]
which is independent of the optimal workloads. We let $\eta_1^{c} = \min_{j \in \BS} b_j/({\tau_{1j} \lambda_1})$ to be the critical step-size that meets the stability condition exactly.

We consider symmetric backends with $a_j=1$ and $b_j=2$, and frontend with an arrival rate of $\lambda_1 = 1$. We try long delays ($\tau_{1j}=1$) and short delays ($\tau_{1j}=0.1$). The critical step-sizes are $\eta_1^c = 5$ and $\eta_1^c = 0.5$, respectively. For each case, we choose small step-sizes that satisfy condition~\eqref{eq:stability-condition1} and large step-sizes that do not satisfy condition~\eqref{eq:stability-condition1}. 

We plot the evolution of workloads and routing probabilities in Figure~\ref{fig:simple-model-path} for the different combinations of delays and step-sizes. We present results for initial workloads of zero and routing probabilities equal to $(.1, .9)$. Similar behavior is observed for other initial conditions.

Delays introduce oscillatory behavior in {\DLBshort}'s decision variables with the period of the oscillations being larger with longer delays. We see that step-sizes lower than the critical values lead to convergence, with the amplitude of the oscillations decreasing monotonically over time. For step-sizes higher than the critical values, the amplitude of the oscillations is large and does not decrease over time. Interestingly, the routing proportions vary widely for large step-sizes, even reaching the boundaries of the probability simplex in some cases. This experiment illustrates that the single-frontend stability condition \eqref{eq:stability-condition1} is nearly tight for this particular example.

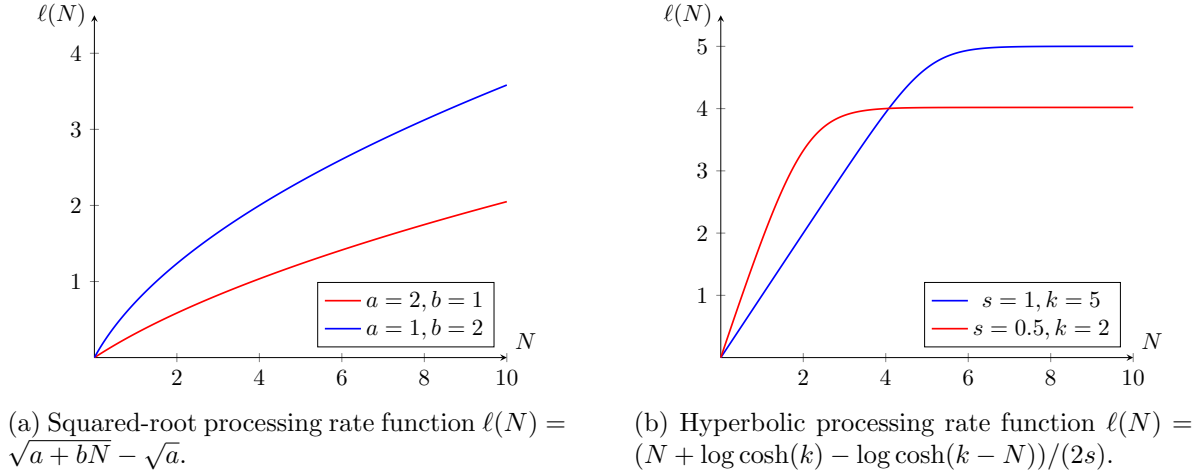
\begin{figure}[t]
    \centering
    \begin{subfigure}[t]{0.45\textwidth}
    \centering
    \scalebox{0.8}{ 
    \begin{tikzpicture}[every node/.style={scale=1}] 
        \begin{axis}[axis lines = middle, xmin = 0, xmax = 10, ymin=0, ymax=4.5, xlabel=$N$,ylabel=$\ell(N)$, legend pos = south east, legend entries = {$a=2,b=1$\\$a=1, b= 2$\\}, xlabel style={above right}, ylabel style={left}]
        \newcommand{\pa}{2}
        \newcommand{\pb}{1}        \addplot[domain=0:10,samples=500,thick,red](x,{sqrt(\pa + \pb*x) - sqrt(\pa)});
        \renewcommand{\pa}{1}
        \renewcommand{\pb}{2}  
        \addplot[domain=0:10,samples=500,thick,blue](x,{sqrt(\pa + \pb*x) - sqrt(\pa)});
        \end{axis}
    \end{tikzpicture}
    }
    \caption{Squared-root processing rate function $\ell(N) = \sqrt{a + bN} - \sqrt{a}$.}
    \label{fig:processing-rate-function1}
    \end{subfigure}%
    \hspace{2em}
    \begin{subfigure}[t]{0.45\textwidth}
    \centering
    \scalebox{0.8}{
    \begin{tikzpicture}[every node/.style={scale=1}]
        \begin{axis}[axis lines = middle, xmin = 0, xmax = 10, ymin=0, ymax=5.5, xlabel=$N$,ylabel=$\ell(N)$, legend pos = south east, legend entries = {$s=1,k=5$\\$s=0.5, k= 2$\\}, xlabel style={above right}, ylabel style={left}]
        \newcommand{\ps}{1}
        \newcommand{\pk}{5}
        \newcommand{\ee}{2.71828}
        \addplot[domain=0:10,samples=500,thick,blue]     
        {(2*x+ln((\ee)^(2*\pk)+1)-ln((\ee)^(2*\pk)+(\ee)^(2*x)))/(2*\ps)};
        \renewcommand{\ps}{0.5}
        \renewcommand{\pk}{2}
        \addplot[domain=0:10,samples=500,thick,red]     
        {(2*x+ln((\ee)^(2*\pk)+1)-ln((\ee)^(2*\pk)+(\ee)^(2*x)))/(2*\ps)};
        \end{axis}
    \end{tikzpicture}
    }
    \caption{Hyperbolic processing rate function $\ell(N) = (N+\log\cosh(k) - \log\cosh(k - N))/(2s)$.}
    \label{fig:processing-rate-function2}
    \end{subfigure}
    \caption{Plots of processing rate functions used in our numerical experiments.}
\end{figure}

\begin{figure}[t!]
    \centering
    \subcaptionbox{$\eta = 0.4$ and $\tau = 1$}{\includegraphics[width=0.4\linewidth]{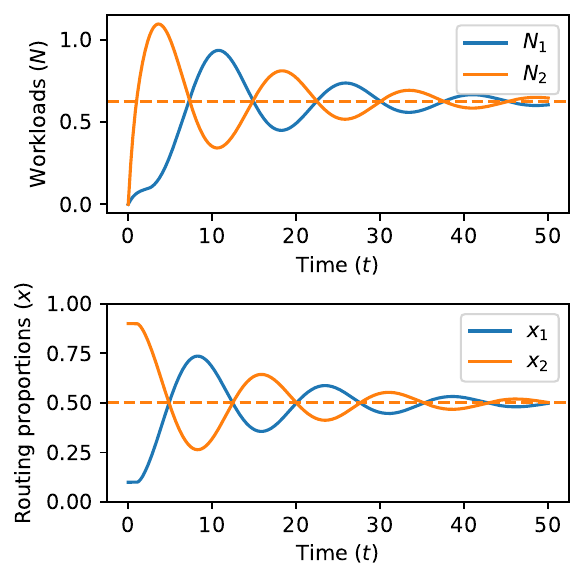}}
    \subcaptionbox{$\eta = 0.6$ and $\tau = 1$}{\includegraphics[width=0.4\linewidth]{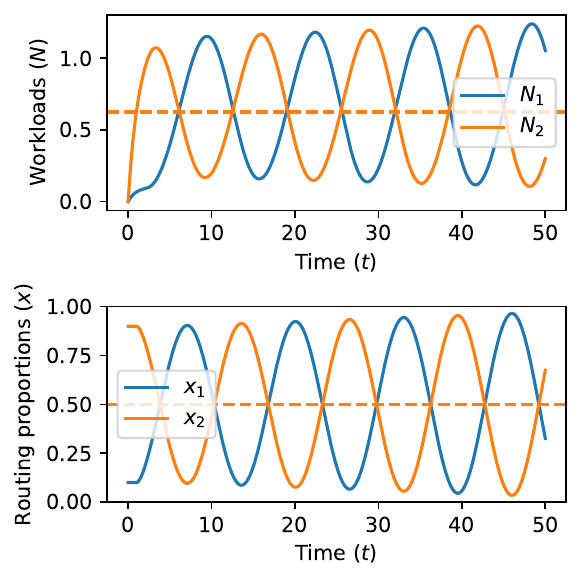}}
    \\
    \subcaptionbox{$\eta = 4$ and $\tau = 0.1$}{\includegraphics[width=0.4\linewidth]{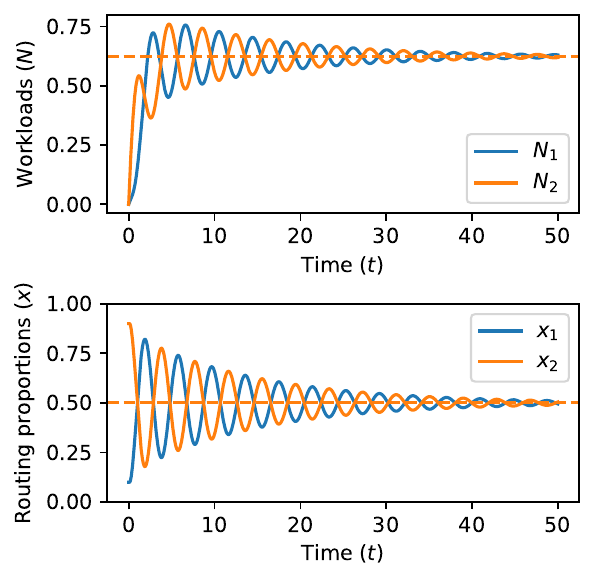}}
    \subcaptionbox{$\eta = 6$ and $\tau = 0.1$}{\includegraphics[width=0.4\linewidth]{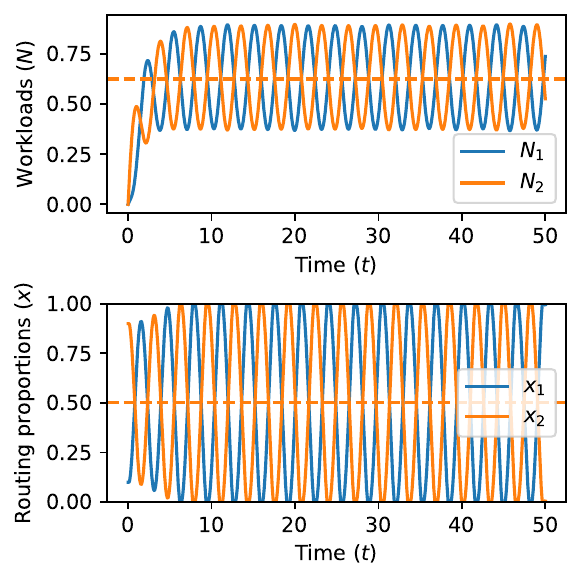}}
    \caption{Evolution of workloads and routing probabilitites as a function of time. Solids lines represent the decision variables of {\DLBshort} and dashed lines represent the optimal solutions of the static routing problem. The top panels have long delays ($\tau=1$) and the bottom panels have short delays ($\tau=0.1$). The algorithm has lower step-sizes and is stable in the left panels and has higher step-sizes and is unstable in the right panels.}
    \label{fig:simple-model-path}
\end{figure}

\subsection{Local Stability Across Different Topologies}

\newcommand{\taumax}{\tau^{\max}}

In this set of experiments, we vary (i) the size of the network and explore imbalanced networks with many frontends and few backends and vice-versa, (ii)  network latencies, (iii) the degree of heterogeneity between frontends arrival rates and backends processing rates.

We generate random complete networks as follows. The number of frontends is $\max(1, \mathrm{Poisson}(\mu_F))$ where $\mathrm{Poission}(\mu_F)$ is a Poisson random variable with mean $\mu_F$, while the number of backends is $\max(2, \mathrm{Poisson}(\mu_B))$ to have at least two backends. We study small networks with $\mu_F = \mu_B = 2$ and larger networks with $\mu_F = \mu_B = 5$.

Backend $j \in \BS$ has $k_j = \max(1, \mathrm{Poisson}(5))$ servers, each taking $s_j$ seconds to process a request, where $s_j$ is lognormally distributed with a expected value of 1 second. The processing rate function of backend $j$ has the hyperbolic form $\ell_j(N_j) = (N_j+\log\cosh(k_j) - \log\cosh(k_j - N_j))/(2s_j)$. The processing rate function is approximately linear with rate $1/s_j$ when the workloads are smaller than the number of servers $k_j$ and reaches a plateau around $k_j$ as capacity gets saturated. We illustrate the hyperbolic processing rate functions in Figure~\ref{fig:processing-rate-function2}.


For network latencies, we place frontends and backends at random points in the unit sphere and then compute distances between frontends and backends using the great circle distance. Denoting by $d_{ij}\ge0$ the distance in the unit sphere between frontend $i \in \FS$ and backend $j \in \BS$, we set $\tau_{ij} = d_{ij} / \pi \cdot \taumax$ where $\taumax$ is the maximum possible network latency and $\pi$ is Archimedes' constant. Because great circle distances in the unit sphere are at most $\pi$, we have that $\tau_{ij} \in [0, \taumax]$. To measure the impact of network latencies, we try $\taumax \in \{0.1, 1\}$, corresponding to low and high network latencies, respectively. In the first case, network latencies are roughly 10\% of the typical serving latency without congestion, and in the second case, network latencies are comparable to serving latencies without congestion.

For the frontend arrival rates, we first pick a vector $\boldsymbol{y}$ in the $|\FS|$-dimensional simplex and then set $\lambda_i = y_i \cdot \rho \cdot \sum_{b \in \BS} \ell_b(\infty)$ where $\rho = 0.9$ is the utilization of the network.


For the step-sizes, we first compute critical step-sizes $\eta_i^{c}$ for each frontend $i \in \FS$ that make the stability condition \eqref{eq:stability-general} hold with equality. The step-sizes are chosen to be proportional to arrival rates, i.e., $\eta_i^{c} / \lambda_i$ is constant across frontends. To do so, we leverage that the left-hand side of \eqref{eq:stability-general} is a positively homogeneous function of step-sizes. We then set the step-sizes to be $\eta_i = \alpha \cdot \eta_i^{c}$, where we try step-size multipliers $\alpha \in \{0.5, 2\}$. For values $\alpha<1$ condition \eqref{eq:stability-condition1} is satisfied while for values $\alpha>1$ the condition is violated. 

For each combination of parameters, we generate 10 random instances and run the simulation for $T=100$ seconds. The processing rate functions are flat for large values of the workload, which can lead to gradients that are too large. To avoid overflows, we clip the gradients of frontend $i\in\FS$ to be at most $4 c_i$, where $c_i$ is the optimal Lagrange multiplier. To measure the impact of perturbations on local stability, the initial routing probabilities are set to $\bx(t) = 0.9 \cdot \bx^* + 0.1 \cdot \bx^{\mathrm{RANDOM}}$ for all $t\le 0$ where $\bx^*$ are the optimal routing probabilities and $\bx_i^{\mathrm{RANDOM}}$ is a drawn uniformly at random from the probability simplex $\Delta_i$ for each frontend $i\in\FS$, and the initial workloads are set to $\bN(t) = 0.9 \cdot \bN^* + 0.1 \cdot \bN^{\mathrm{RANDOM}}$ for $t\le0$ where $\bN^*$ are the optimal workloads and $N_j^{\mathrm{RANDOM}}$ is a drawn uniformly from $[0, 2 k_j]$ for each backend $j \in \BS$.

We measure the gap of the algorithm as the relative difference between average amount of request in our algorithm and the static routing problem, i.e., 
\begin{align}\label{eq:GAP}
    \operatorname{GAP} = \frac 1 {\operatorname{OPT}} \frac 1 T \int_{0}^T \left( \sum_{j \in \BS} N_j(t) + \sum_{(i,j) \in \AS} N_{ij}(t) \right) - 1\,.    
\end{align}
Note that the gap can be negative if the initial workloads are low and {\DLBshort} quickly converges to the optimal values. 

We also measure convergence to the optimal solution by looking at the distance to optimality, averaged over the last $4 \taumax$ seconds. Averaging leads to a better representation of convergence trends as oscillations can lead to widely changing distances to optimality. We measure convergence of the workloads and routing probabilities as follows
\[
    \operatorname{error}_N = \frac 1 {4\taumax} \int_{T -4\taumax}^{T} \| \bN(t) - \bN^*\|_2 dt \quad \text{and} \quad \operatorname{error}_x = \frac 1 {4\taumax} \int_{T -4\taumax}^{T} \| \bx(t) - \bx^*\|_2 dt\,.
\]
Finally, we visually inspect the evolutions of workloads generated by {\DLBshort} and record the fraction of instances in which they converge to the optimal levels.

We report results in Table~\ref{tab:results}, where we present averages over all 10 instances for each combination of parameters. We see that when our sufficient condition \eqref{eq:stability-general} is satisfied (step-size multiplier $\alpha=0.5$), {\DLBshort} converges to an optimal solution when starting close to it. When our sufficient condition is not satisfied (step-size multiplier $\alpha=2$), {\DLBshort} can fail to converge. Moreover, our condition is not necessary as the algorithm converges even when the condition is not satisfied. Finally, averages GAPs and distances of the workloads to the optimal levels ($\operatorname{error}_N$) are usually small when step-sizes satisfy the stability condition. We observe, however, that the distance of the routing probabilities to the optimal ones ($\operatorname{error}_x$) does not always converge to zero even when workloads converge to optimality. In those cases, networks have cycles whose travel times sum to close to zero and there is a large set of solutions that are nearly optimal. Non-uniqueness, however, does not seem to impact convergence of the objective value. Errors for small and large networks are similar, but GAPs are smaller for smaller networks, which suggests than {\DLBshort} can take longer to converge in larger networks.

\begin{table}
\small
\begin{tabular}{ccccccc}
& Step-size & Satisfies \\
Network parameters  & multiplier & stability cond. \eqref{eq:stability-general} & 
$\operatorname{GAP}$ & $\operatorname{error}_N$  & $\operatorname{error}_x$ & Converged? \\
\hline 
$(\mu_F=2, \mu_B=2, \taumax=0.1)$ & $\alpha=0.5$ & Yes & 0.09\%  & 0.00276 
 & 0.0247 & 100\% \\
$(\mu_F=2, \mu_B=2, \taumax=0.1)$ & $\alpha=2$ & No & 105\%
 & 40.8 & 0.391 & 30\% \\
$(\mu_F=2, \mu_B=2, \taumax=1)$ & $\alpha=0.5$ & Yes & 0.13\%
 & 0.00314 &	0.0289 & 100\% \\
$(\mu_F=2, \mu_B=2, \taumax=1)$ & $\alpha=2$ & No & 57.3\% & 26.80 & 0.271 & 40\% \\
$(\mu_F=5, \mu_B=5, \taumax=0.1)$ & $\alpha=0.5$ & Yes & 0.32\% & 0.0170 & 0.141 & 100\% \\
$(\mu_F=5, \mu_B=5, \taumax=0.1)$ & $\alpha=2$ & No & 2.95\% & 0.933 & 0.150 & 80\% \\
$(\mu_F=5, \mu_B=5, \taumax=1)$ & $\alpha=0.5$ & Yes & 1.03\% & 0.252 & 0.144 & 100\% \\
$(\mu_F=5, \mu_B=5, \taumax=1)$ & $\alpha=2$ & No & 4.18\% & 4.39 & 0.087 & 70\% \\
\end{tabular}
\caption{Average performance of {\DLBshort} for different combination of parameters when initial state is close to optimality.}\label{tab:results}
\end{table}

\subsection{Global Stability and Benchmarking}

In this final set of experiments, we study the global stability of {\DLBshort} and compare our algorithms against the following load balancing algorithms:
\begin{itemize}
    \item \emph{Least workload} (LW), which routes an incoming request to the connected backend with lowest workload, i.e., at time $t$ frontend $i$ picks a backend $j \in \BS_i$ with the lowest $N_j(t - \tau_{ij})$.
    
    \item \emph{Least latency} (LL), which routes an incoming request to the connected backend with lowest network and serving latency, i.e., at time $t$ frontend $i$ picks a backend $j \in \BS_i$ with the lowest $\tau_{ij} + L_{j}(N_j(t - \tau_{ij}))$ where $L_j(N) = N / \ell_j(N)$ is as estimate for the serving latency of backend $j$ when the workload is $N$.
    
    \item \emph{Greatest Marginal Service Rate} (GMSR), which routes requests to the connected backends with the largest greatest marginal service rate, i.e., at time $t$ frontend $i$ picks a backend $j \in \BS_i$ with the largest $\ell_j'(N_j(t - \tau_{ij}))$. 
\end{itemize}

We remark that policies such as Join-the-Shortest-Queue (JSQ) and Join-the-Idle-Queue (JIQ) cannot be directly implemented in our setting because there is no notion of queue in our model. Nevertheless, LW is similar in spirit to JSQ or MaxPressure be directly implemented in our setting because there is no notion of queue in our model.as it chooses the least congested backend while LL is similar to the Never-Queue policy of \citet{shenker1989optimal}, which chooses the fastest server of the idle ones.

We use the random networks from the second set of experiments but choose the initial routing probabilities and workloads to be random points not necessarily close to the optimal ones. That is, we set $\bx(t) = \bx^{\mathrm{RANDOM}}$ for all $t\le 0$ and $\bN(t) = \bN^{\mathrm{RANDOM}}$ for $t\le0$, with $\bx^{\mathrm{RANDOM}}$ and $\bN^{\mathrm{RANDOM}}$ chosen as before. Moreover, we run the simulation for $T=1000$ seconds. For {\DLBshort}, we try step-size multipliers $\alpha = \{0.01, 0.05, 0.1, 0.5\}$ and report the best multiplier for each instance. For each instance, we run all policies and compute the gap relative to optimal static routing problem \eqref{eq:GAP} but averaged over the last $4 \taumax$ seconds and the distance to the optimal workloads $\operatorname{error}_N$. Results are presented in Table~\ref{tab:results-benchmarks}, where we present averages over all 10 instances for each combination of parameters.

Our results confirm the global stability of {\DLBshort} for arbitrary initial states. In all examples, we visually verify that the routing probabilities and workloads of {\DLBshort} converge to optimal values. When choosing the step-sizes we observe a tradeoff between convergence and responsiveness. Low step-sizes guarantee convergence to optimal values at the expense of longer convergence times. When initial routing probabilities are very suboptimal (e.g., we initially route a large amount of requests to a small backend), workloads can take long excursions before attaining optimal values. When step-sizes are large, {\DLBshort} reaches quickly the orbit of an optimal solution, but it can perpetually oscillate around optimal values.

The load balancing LW, LL, and GMSR do not converge to optimal values, and their performance relative to the optimal static routing is poor. Among these, LL performs best, but, still, its performance is orders of magnitude worse than {\DLBshort}. Because of feedback delays, the workloads under these policies oscillate widely, with the amplitude of the oscillations increasing proportionally with the network latencies. Intuitively, these policies are too reactive and delayed information leads to oscillatory behavior. For example, LL aims to balance the backends' workloads by routing requests to the backend with the least workload. When a frontend learns that backends' workloads are equal, the system might have already moved beyond the equilibrium in the opposite direction. This overcorrection can then trigger a response in the other direction, again based on delayed information, leading to oscillations around the equilibrium.

\begin{table}
\small
\begin{tabular}{ccccccccc}
& \multicolumn{2}{c}{{\DLBshort}} & \multicolumn{2}{c}{LW} & \multicolumn{2}{c}{LL} & \multicolumn{2}{c}{GMSR} \\
Network parameters  & $\operatorname{GAP}$ & $\operatorname{error}_N$  & $\operatorname{GAP}$ & $\operatorname{error}_N$& $\operatorname{GAP}$ & $\operatorname{error}_N$& $\operatorname{GAP}$ & $\operatorname{error}_N$ \\
\hline 
$(\mu_F=2, \mu_B=2, \taumax=0.1)$ & 0.057\% & 0.000340 & 38.4\% & 5.38 & 40.1\% & 6.79 & 19.4\% & 2.84\\ 
$(\mu_F=2, \mu_B=2, \taumax=1)$ & 0.17\% & 0.00145 & 136\% & 37.5 & 92.9\% & 26.6 & 170\% & 56.8 \\
$(\mu_F=5, \mu_B=5, \taumax=0.1)$ & 0.54\% & 0.0144 & 73.3\% & 11.3 & 129\% & 24.7 & 42.3\% & 7.44\\
$(\mu_F=5, \mu_B=5, \taumax=1)$ & 2.51\% & 0.347 & 260\% & 75.6 & 173\% & 61.4 & 380\% & 156
\end{tabular}
\caption{Average for different combination of parameters for random initial states and different policies. LW stands for least workload, LL stands for least latency, and GMSR stands for greatest marginal service rate.}\label{tab:results-benchmarks}
\end{table}

\section{Conclusions}

In this paper, we introduce {\DLBshort}, a distributed algorithm for load balancing in bipartite networks with network latencies based on decentralized gradient descent. We prove our algorithm is locally asymptotically stable and provide sufficient conditions for stability on the step-size of gradient descent. 

Several interesting research directions stem from this work. First, providing sufficient conditions for global stability would be a natural research direction. Our experiments suggest that {\DLBshort} can be globally stable, but step-size tuning is critical to guarantee convergence. Second, our results focus on a fluid model in which dynamics are deterministic. In practice, arrivals and processing times are random, and it would be interesting to study the performance of {\DLBshort} using a discrete, stochastic model. Finally, we assumed knowledge of the processing rate gradients, which need to be estimated. Another direction is to design a data-driven version of {\DLBshort} that learns gradients on the fly.


\bibliographystyle{plainnat}
\bibliography{references}

\newpage
\appendix

\section{Proof of Main Result for General Networks}\label{sec:general-case}

We now present our analysis for general networks. We need to analyze the loop transfer function \eqref{eq:hat-loop}. Recall that because the constants $c_i$ are different across frontends, the matrix $\sum_{i \in \mathcal F}
\exp(-2sc_i) \lambda_i \eta_i E_i$ in the loop transfer function is not necessarily positive semi-definite and we cannot simply decompose the loop function $\hat L(s)$ as a product of positive semi-definite matrix and a diagonal matrix. We analyze the loop function using a uniformization technique that decomposes the matrix into a positive semi-definite matrix and a ``small'' perturbation. This technique is inspired by \citet{han2006multi}, but our analysis is more complex because of the summation over frontends, which introduces multiple terms. Let $\hat c \in \mathbb R_+$ be a ```pivot'' to be optimized later. We write
\[
    \sum_{i \in \mathcal F}
\exp(-2sc_i) \lambda_i \eta_i E_i = \exp(-2s\hat c) \Bigg(\underbrace{\sum_{i \in \mathcal F} \lambda_i \eta_i E_i}_{\text{p.s.d.}} +  \underbrace{\sum_{i \in \mathcal F}
\Big(\exp\big(2s(\hat c - c_i)\big) - 1\Big) \lambda_i \eta_i E_i}_{\text{perturbation}} \Bigg)\,.
\]
The following lemma provides some intuition for our decomposition.
\begin{lem}\label{lem:exp-bound} For every $a \in \mathbb R$ we have that $|\exp(\im a) - 1| \le |a|\,.$
\end{lem}

For low frequencies $w$, Lemma~\ref{lem:exp-bound} implies that the perturbation term is small and the approximation is of good quality when $s = \im w$. For high frequencies $w$, however, the perturbation term is large but $D(\im w)$ has norm of order $|w|^{-2}$ and, thus, all eigenvalues of the loop function are close to zero. We require a delicate analysis to appropriately balance these trade-offs.

The following linear algebra lemma provides a key ingredient of our analysis as it allows to analyze the spectrum of the loop function in terms of the positive semi-definite and perturbation terms of our decomposition.

\begin{lem}\label{lem:spectrum-product-sum}
Let $A$ be a positive semi-definite singular matrix with pseudo-inverse $A^+$. If $\operatorname{range}(P) \subseteq \operatorname{kernel}(A)^\bot$, then
\[
    \operatorname{spec}\left( (A+P) B\right) \subseteq W(A) \left( W(B) + W(A^+ P B) \right)\,.
\]
\end{lem}

To invoke the previous lemma, we need an more refined characterization of the Laplacian matrices $E_i$ and their weighted sums.

\begin{lem}\label{lem:laplacian-matrix2}
The following holds:
\begin{enumerate}
    \item For each frontend $i \in \FS$, we have $\operatorname{kernel}(E_i) = \left\{ \by \in \mathbb C^{|\BS|} : y_j = y_{j'}\quad \forall j, j' \in \BS_i \right\}$.
    \item We have $\operatorname{kernel}\left(\sum_{i\in\FS} \lambda_i \eta_i E_i\right) = \operatorname{span}(\mathbf 1)$ if the network $\mathcal G$ is connected.
    \item The spectral gap or minimum non-zero eigenvalue of $\sum_{i\in\FS}\lambda_i \eta_i E_i$ when the network $\mathcal G$ is connected is given by
    \[
    \operatorname{gap}\left(\sum_{i\in\FS}\lambda_i \eta_i  E_i\right) = \min_{\by \in \mathbb R^{|\BS|} : \mathbf 1^\top \by = 0} \frac{\sum_{i\in\FS}\lambda_i \eta_i  \by^\top E_i \by} {\by^\top \by}\,.
    \]
\end{enumerate}
\end{lem}

We invoke Lemma~\ref{lem:spectrum-product-sum} with $B = D(s)$, $P = \sum_{i \in \mathcal F}
\left(\exp\big(2s(\hat c - c_i)\big) - 1\right) \lambda_i \eta_i E_i$, and $A = \sum_{i \in \mathcal F} \lambda_i \eta_i E_i$. The matrix $A$ is positive semi-definite by Lemma~\ref{lem:laplacian-matrix1} because the conic sum of positive semi-definite matrices is positive semi-definite and $\lambda_i \eta_i \ge 0$. To see that $\operatorname{range}(P) \subseteq \operatorname{kernel}(A)^\bot$ it suffices to check that $\mathbf 1^\top P = 0$ because $\operatorname{kernel}(A)^\bot = \left\{ \by \in \mathbb C^{|\BS|} : \mathbf 1^\top \by = 0 \right\}$ by Lemma~\ref{lem:laplacian-matrix2}. Therefore,
\[
    \mathbf 1^\top P = \sum_{i \in \mathcal F}
\left(\exp\big(2s(\hat c - c_i)\big) - 1\right) \lambda_i \eta_i \mathbf 1^\top E_i = 0\,,
\]
where the last equation follows because $\mathbf 1 \in \operatorname{kernel}(E_i^\top) = \operatorname{kernel}(E_i)$ by  Lemma~\ref{lem:laplacian-matrix2} together with the fact that the matrix $E_i$ is symmetric.

Because $\hat L(\im w)$ is conjugate symmetric, i.e., $\hat L(\im w) = \overline{\hat L(-\im w)}$, we only need to check the loop function does not encircle $-1 + \im 0$ for $w \in [0, \infty)$ as results for negative values follow by symmetry. Therefore, we have that
\begin{align*}
    \operatorname{spec}(\hat L(\im w)) &
    =\operatorname{spec}\left( \left( \sum_{i \in \mathcal F}
    \exp(-\im 2 w c_i) \lambda_i \eta_i E_i \right) D(\im w)\right)\\
    &= \exp(-\im 2 w \hat c) \cdot \operatorname{spec}\Bigg(  \Bigg( \underbrace{\sum_{i \in \mathcal F} \lambda_i \eta_i E_i}_{=:A} +  \underbrace{\sum_{i \in \mathcal F}
\Big(\exp\big(\im 2 w(\hat c - c_i)\big) - 1\Big) \lambda_i \eta_i E_i}_{=:P} \Bigg) D(\im w) \Bigg) \\
    &  \subseteq \exp(-\im 2 w \hat c)\cdot W\left(A\right) \cdot \left( W(D(\im w)) + W\left(A^+ P D(\im w)\right) \right)\\
    &  \subseteq  \exp(-\im 2 w \hat c)\cdot[0,\boldsymbol{\lambda}^\top \boldsymbol{\eta}] \cdot W(D(\im w)) + \exp(-\im 2 w \hat c) \cdot [0,\boldsymbol{\lambda}^\top \boldsymbol{\eta}] \cdot W\left(A^+ P D(\im w)\right)\\
    &\subseteq [0,\boldsymbol{\lambda}^\top \boldsymbol{\eta} \exp(-\im 2 w \hat c) ] \cdot W(D(\im w)) +  \operatorname{disk}\left(\boldsymbol{\lambda}^\top \boldsymbol{\eta} \| A^+ P D(\im w) \| \right)\,,
\end{align*}
where the second equation follows because the spectrum preserves scalar multiplications; the first inclusion follows from Lemma~\ref{lem:spectrum-product-sum}; the second inclusion because $W(A) = W\left(\sum_{i \in \mathcal F} \lambda_i \eta_i E_i\right) \subseteq \sum_{i \in \mathcal F} \lambda_i \eta_i W(E_i) = \sum_{i \in \mathcal F} \lambda_i \eta_i [0,1] = [0,\boldsymbol{\lambda}^\top \boldsymbol{\eta}]$ by sub-additivity of the numerical range and Lemma~\ref{lem:laplacian-matrix1}; and the third inclusion because the numerical range of a matrix is contained in a disk of radius equal to its norm together with the fact that $|\exp(-\im 2 w \hat c)|\le 1$ because $\hat c$ and $w$ are real numbers. We next analyze each term at a time.

For the first term, we first let $\hat \tau_j$ be the uniform latency of backend $j$, which is given by
\begin{align*}
    \hat \tau_j = \hat c - \frac{1} {\ell_j'}\,.
\end{align*}
This definition naturally puts a constraint on the pivot $\hat c$. We need $\hat c \ge 1/\ell_j'$ for all $j \in \BS$ for the uniform latencies to be non-negative. Using this definition we can obtain an expression that is amenable to invoking Lemma~\ref{lem:geometric-bounds}. We have that
\begin{align*}
     [0,\boldsymbol{\lambda}^\top \boldsymbol{\eta} \exp(-\im 2 w \hat c) ] \cdot W(D(\im w)) &= \operatorname{conv}\left(0, \frac{\boldsymbol{\lambda}^\top \boldsymbol{\eta} \sigma_j e^{-\im 2w\left(\hat c - 1/\ell_j'\right)} }{(\im w)^2 + \im w\ell_j'} : j \in \BS \right)\\
     &=\operatorname{conv}\left(0, \frac{2\hat \tau_j \boldsymbol{\lambda}^\top \boldsymbol{\eta} \sigma_j}{\ell_j'} \frac{e^{-\im 2 w \hat\tau_j} }{2 \hat \tau_j \im w(\im w(\hat c - \hat \tau_j) + 1)} : j \in \BS \right)\\
     &\subseteq \max_{j \in \BS} \frac{2\hat \tau_j \boldsymbol{\lambda}^\top \boldsymbol{\eta} \sigma_j}{\ell_j'} \cdot \operatorname{conv}\left(0,  \frac{e^{-\im 2w \hat\tau_j} }{2 \hat \tau_j \im w(\im w(\hat c - \hat \tau_j) + 1)} : j \in \BS \right)\\
      &\subseteq \max_{j \in \BS} \frac{2\hat \tau_j \boldsymbol{\lambda}^\top \boldsymbol{\eta} \sigma_j}{\ell_j'} \cdot \hat{\mathcal H}\,,
\end{align*}
where the first equation follows because the numerical range of a diagonal matrix is the convex hull of its entries and convex hulls preserve scalar multiplications, the first inclusion from extracting the maximum factor, and the last from Lemma~\ref{lem:geometric-bounds} and denoting by $\hat{\mathcal H}$ the lower halfspace with slope $1/(w\hat c)$ that goes through $-1 + \im 0$, i.e.,
\[
\hat{\mathcal H} = \left\{  z \in \mathbb C : \Re(z) \ge -1 + w \hat c\cdot \Im(z) \right\}\,.
\]

For the second term, we use that $\| A^+ P D(\im w) \| \le \|A^+\| \cdot \|P\| \cdot \|D(\im w) \|$ and bound each matrix independently. For the first matrix, we use that the norm of the pseudo-inverse is the reciprocal of the spectral gap to obtain that
\[
    \| A^+\| = \left\| \left( \sum_{i\in\FS}\lambda_i \eta_i E_i\right)^+ \right\| = \operatorname{gap}\left(\sum_{i\in\FS}\lambda_i \eta_i E_i\right)^{-1}\,.
\]
In the following, we use $\operatorname{gap}$ as shorthand for the spectral gap of the Laplacian matrix $\sum_{i\in\FS}\lambda_i \eta_i E_i$. For the perturbation matrix, we have
\begin{align*}
    \| P \| &= \left\| \sum_{i \in \mathcal F}
\Big(\exp\big(\im 2 w(\hat c - c_i)\big) - 1\Big) \lambda_i \eta_i E_i \right\|\\
&\le \sum_{i \in \mathcal F}
\Big|\exp\big(\im 2 w(\hat c - c_i)\big) - 1\Big| \cdot \lambda_i \eta_i \cdot \| E_i \|\\
&\le2 w \sum_{i \in \FS} \lambda_i \eta_i |\hat c - c_i|\,,
\end{align*}
where the first inequality follows from the triangle inequality and positive homogeneity properties of the norm; and the second inequality by Lemma~\ref{lem:exp-bound}, using that $w \ge 0$, together with $\|E_i\| \le 1$ by Lemma~\ref{lem:laplacian-matrix1}. For the diagonal matrix, we have
\[
    \| D(\im w) \| \le \max_{j \in \BS} \left| \frac{\sigma_j e^{ \im 2w/\ell_j'} }{(\im w)^2 + \im w\ell_j'} \right| \le \max_{j \in \BS} \frac{\sigma_j }{w \sqrt{w^2 + (\ell_j')^2}}\,.
\]
Combining these factors, we obtain that the second term can be bounded as follows
\[
    \| A^+ P D(\im w) \| \le \frac{2  \sum_{i \in \FS} \lambda_i \eta_i  |\hat c - c_i|}{\operatorname{gap}} \cdot \max_{j \in \BS} \frac{\sigma_j }{\sqrt{w^2 + (\ell_j')^2}}\,.
\]

Putting everything together, we obtain that
\[
     \operatorname{spec}(\hat L(\im w)) \subseteq
    \underbrace{\max_{j \in \BS} \frac{2\hat \tau_j \boldsymbol{\lambda}^\top \boldsymbol{\eta} \sigma_j}{\ell_j'}}_{:=a} \cdot \hat{\mathcal H} + \operatorname{disk}\Bigg(\underbrace{\frac{2 \boldsymbol{\lambda}^\top \boldsymbol{\eta}  \sum_{i \in \FS} \lambda_i |\hat c - c_i|}{\operatorname{gap}} \cdot \max_{j \in \BS} \frac{\sigma_j }{\sqrt{w^2 + (\ell_j')^2}}}_{:=b(w)}\Bigg)\,,
\]
A sufficient condition for stability is that the distance from the scaled halfspace $a \cdot \hat{\mathcal H}$ to $-1 + 0 \im$ is larger than the radius of the disk $b(w)$, i.e.,
\begin{align*}
    b(w) < \operatorname{dist}(-1, a \cdot \hat{\mathcal H}) = \frac{1-a}{\sqrt{1+w^2\hat c^2}}\,,
\end{align*}
where the last equation follows from using the standard formula for the distance of the point $-1+ 0 \im$ to the line $\left\{  z \in \mathbb C : \Re(z) = -a + w \hat c\cdot \Im(z) \right\}$ defining the boundary of $a \cdot \hat{\mathcal H}$. Therefore, a sufficient condition for stability is
\begin{align}\label{eq:a-b-condition}
    a + b(w) \cdot \sqrt{1+w^2\hat c^2} < 1\,,
\end{align}
which can be further simplified using the following lemma.

\begin{lem}\label{lem:ratio-bound}
For $a \ge 0$, $b > 0$, and $w\ge0$, we have
\[
    \sqrt{\frac{w^2 + a^2}{w^2 + b^2}} \le \max\left(1,\frac{a}{b}\right)\,.
\]
\end{lem}

Using Lemma~\ref{lem:ratio-bound}, we obtain that
\begin{align}\label{eq:almost-final}
    &\max_{j \in \BS} \frac{\sigma_j }{\sqrt{w^2 + (\ell_j')^2}} \cdot \sqrt{1+w^2\hat c^2}
    = \hat c \cdot \max_{j \in \BS} \frac{\sigma_j \cdot \sqrt{w^2 + (1/\hat c)^2} }{\sqrt{w^2 + (\ell_j')^2}} \nonumber \\
    &\quad\le \hat c \cdot \max_{j \in \BS} \sigma_j \cdot \max\left(1,\frac{1}{\hat c \ell_j'}\right)
    = \hat c \cdot \max_{j \in \BS} \sigma_j \,,
\end{align}
where the last equation follows because $\hat c \ell_j' \ge 1$ for all $j \in \BS$. We therefore conclude the proof of Theorem~\ref{thm:stability-general} by combining \eqref{eq:a-b-condition} and \eqref{eq:almost-final}.

\section{Projection Algorithm}\label{app:projection}

Algorithm~\ref{alg:projection} presents an efficient procedure to find the projection of a vector $\boldsymbol{z} \in \mathbb R^n$ to the tangent cone of the probability simplex at a vector $\bx \in \Delta$. That is, we aim to solve the Euclidean norm projection problem
\[
\Pi_{T_{\Delta}(\bx)}(\boldsymbol{z}) = \arg\min_{\bv \in T_{\Delta}(\bx)} \|\bv - \boldsymbol{z}\|_2\,,
\]
where the tangent cone at is
\[
 T_{\Delta}(\bx) = \left\{\bv \in \mathbb{R}^{n}: \sum_{j\in [n]} v_j = 0,  v_j \geq 0  \text{ if } x_{i}  = 0\right\}\,.
\]
The algorithm's complexity is $O(n \log n)$, which is governed by the time taken to sort the components of $\boldsymbol{z}$ associated with $x_i = 0$. The next results establishes the correctness of our algorithm.

\begin{lem}
For every $\boldsymbol{z} \in \mathbb R^n$ and $\bx \in \Delta$, Algorithm~\ref{alg:projection} returns the projection of $\boldsymbol{z}$ to the tangent cone of the probability simplex at $\bx$.
\end{lem}
\begin{proof}
Lemma~\ref{lem:projection-analysis} presents the Karush-Kuhn-Tucker optimality conditions of the projection problem, which are necessary and sufficient because the objective is convex and continuously differentiable and the constraints linear. 

Let $\bv \in \mathbb R^n$ be the optimal solution, which is unique because the objective is strictly convex. We denote by $\mathcal T = \{ i \in [n] : x_i > 0\}$ the set of indices for which $x_i$ is positive, $\mathcal S^> = \{ i \in [n] : x_i = 0 \text{ and } v_i > 0\}$, and $\mathcal S^0 = \{ i \in [n] : x_i \text{ and } v_i = 0\}$. We have $[n] = \mathcal T \cup \mathcal S^0 \cup \mathcal S^>$. Let $\beta \in \mathbb R$ be the Lagrange multiplier of the constraint $\sum_{i \in [n]} v_i = 0$ and $\mu_i \ge 0$ be the multiplier of the constraint $v_i \ge 0$ for $i \not\in \mathcal T$. The optimality conditions are
\[
    v_i - z_i + \beta = 0\quad\forall i \in \mathcal T\,,
\]
and 
\[
    v_i - z_i + \beta - \mu_i = 0\quad\forall i \not\in \mathcal T\,,
\]
together with the complementary slackness condition for $i \not \in T$ given by $\mu_i v_i = 0$. 

For $i \in \mathcal T$, we do not have any constraint on $v_i$ and we simply set $v_i = z_i - \beta$. For $i \in \mathcal S^>$ we have that $\mu_i = 0$ from complementary slackness and we set $v_i = z_i - \beta$. For feasibility, we require that $v_i \ge 0$ or equivalently $z_i \ge \beta$. For $i \in \mathcal S^0$ we have that $v_0 = 0$, which is trivially feasible, and because $\mu_i \ge 0$ we require that $\beta \ge z_i$. Equation \eqref{eq:beta-formula} implies that the Lagrange multiplier $\beta$ satisfies
\begin{align}\label{eq:beta-formula2}
    \beta = \frac 1 {|\mathcal T| + |\mathcal S^>|} \left( \sum_{j \in \mathcal T} z_j + \sum_{j \in \mathcal S^>} z_j\right)\,.    
\end{align}

We now discuss the correctness of the algorithm. We shall prove that at termination, the set $\mathcal S$ maintained by the algorithm satisfies $\mathcal S = \mathcal S^>$. If so, by construction we would have that $v_i = 0$ for $i \in \mathcal S^0$ because the set $\mathcal S$ has initially all $i \in [n]$ with $x_i = 0$ and in each iteration of the while loop we remove an element $i$ with $x_i = 0$ that we set to $v_i = 0$. The values for $i \in \mathcal T$ and $i\in\mathcal S$ are set at termination to $v_i = z_i - \beta$ with $\beta$ satisfying \eqref{eq:beta-formula2}, which satisfy the first-order conditions. Therefore, to conclude we need to show that (i) $z_i \ge \beta$ for all $i\in \mathcal S$ (the components that remain in $\mathcal S$ at termination) and (ii) $\beta \ge z_i$ for $i \in [n]$ with $x_i = 0$ and $i \not\in \mathcal S$ (the components that are removed from $\mathcal S$ along the run of the algorithm).

First, we show that $z_i \ge \beta$ for all $i\in \mathcal S$. If we never break the while loop, the set $\mathcal S$ is empty and the result is trivial. If we break the while loop, we do so because the $z_i \ge \beta$ for the component $i\in \mathcal S$ with smallest value in $z_i$. Therefore, the result holds for every other larger value in the set $\mathcal S$ because the value of $\beta$ does not change once we break the while loop.

Second, we show $\beta \ge z_i$ for $i \in [n]$ with $x_i = 0$ and $i \not\in S$. If we have a component $i \in [n]$ with $x_i = 0$ and $i \not\in \mathcal S$, then this element should have been removed during one iteration of the while loop at which point we would have that $z_i < \beta$. Note, however, the value of $\beta$ changes as we continue iterating the while loop. The result follows form the observation that $\beta$ increases monotonically during the run of the algorithm because we are always removing the smallest value of $\mathcal S$ and $\beta$ is an average of the values in $\mathcal S$ and $\mathcal T$.
\end{proof}

\begin{algorithm}
\caption{Projection to the tangent cone of the probability simplex}\label{alg:projection}
\begin{algorithmic}[1]
\Require A vector $\boldsymbol{z} \in \mathbb R^n$ to be projected and a vector $\bx \in \Delta$ in the probability simplex.
\Ensure The projection $\bv = \Pi_{T_{\Delta}(\bx)} \left(\boldsymbol{z} \right)$ of the vector $\boldsymbol{z}$ to the tangent cone of $\Delta$ at $\bx$. 

\State $\mathcal S \gets \{ i \in [n] : x_i = 0\}$ and $\mathcal T \gets \{ i \in [n] : x_i > 0\}$.
\State $\bv \gets \mathbf 0$
\While{$\mathcal S \neq \emptyset$}
    \State $\beta \gets \frac 1{|\mathcal T| + |\mathcal S|} \left(\sum_{i \in \mathcal T} z_i + \sum_{i \in \mathcal S} z_i \right)$
    \State Let $i\in \mathcal S$ be the component with the smallest value $z_i$
    \If{$z_i \ge \beta$}
        \Break
    \Else
       \State $v_i = 0$
       \State $\mathcal S \gets \mathcal S \setminus \{i\} $
    \EndIf
\EndWhile
\If{$\mathcal S = \emptyset$}
    \State $\beta \gets \frac 1{|\mathcal T|} \sum_{i \in \mathcal T} z_i$ 
\EndIf
\State $v_i \gets z_i - \beta$ for $i \in \mathcal S \cup \mathcal T$
\State \Return $\bv$
\end{algorithmic}
\end{algorithm}

\section{Proofs of Results}

\subsection{Proof of Lemma~\ref{lem:opt-bound}}

\begin{proof} Consider the following relaxation of $\operatorname{OPT}$, in which the inflow to each backend is at most its outflow:
\begin{equation}\label{eq:opt-relaxed}
\begin{split}
    \operatorname{OPT}' = \min_{\bN,\bx} &\quad \sum_{j \in \BS} N_j + \sum_{(i,j) \in \AS} \lambda_i x_{ij} \tau_{ij} \\
    \text{s.t.}
    & \quad \sum_{i \in \FS_j} \lambda_i x_{ij} \le \ell_j(N_j)\,, \forall j \in \BS\,,\\
    & \quad  \sum_{j \in \BS_i} x_{ij} = 1\,, \quad \forall i \in \FS\,,\\
    & \quad x_{ij} \geq 0\,, \quad \forall (i,j) \in \AS\,.
\end{split}
\end{equation}
Programs \eqref{eq:opt} and \eqref{eq:opt-relaxed} attain the same objective values because, at optimality, the relaxed flow balance constraint of every backend should hold with equality. This follows because, by Assumption~\ref{assume:processing-rate}, processing rate functions are increasing and the objective is increasing in workloads. Therefore, $\operatorname{OPT}' = \operatorname{OPT}$.

Fix a time-average convergent policy and denote by 
\[
    \bar N_j = \lim_{T \rightarrow \infty} \frac 1 T \int_0^T N_j(t) dt\ \quad \text{and} \quad \bar N_{ij} = \lim_{T \rightarrow \infty} \frac 1 T \int_0^T N_{ij}(t) dt
\]
the time-average workloads at the backends and  traveling requests, respectively, which exist by assumption. By Little's Law, we have that time-average routing probabilities 
\[
    \bar x_{ij} = \lim_{T \rightarrow \infty} \frac 1 T \int_0^T x_{ij}(t) dt
\] 
exist and satisfy $\bar N_{ij} = \lambda_i \bar x_{ij} \tau_{ij}$. Therefore, the time-average performance of the policy can be written as $\operatorname{ALG} = \sum_{j \in \BS} \bar N_j + \sum_{(i,j) \in \AS} \lambda_i \bar x_{ij} \tau_{ij}$, which coincides with the objective value of \eqref{eq:opt-relaxed}. To prove that $\operatorname{ALG} \ge \operatorname{OPT}'$, it suffices to show that the solution $(\bar \bN, \bar \bx)$ is feasible for $\operatorname{OPT}'$.

Clearly, we have that $\bar x_{ij} \ge 0$ and $\sum_{j \in \BS_i} \bar x_{ij} =1$. To establish flow balance, we need the following result, which proves that every policy with finite time averages is rate stable.

\begin{lem} If $\limsup_{T \rightarrow \infty} \frac 1 T \int_0^T N_j(t) dt < \infty$, then $\lim_{T\rightarrow\infty} N_j(T)/T = 0$.    
\end{lem}
\begin{proof}
    Because $N_j(T) \ge 0$, it suffices to show that $\limsup_{T\rightarrow\infty} N_j(T)/T \le 0$. Integrating the dynamics of the backends from $t=s>0$ to $t=T$, we obtain that
    \begin{align}\label{eq:lower-bound-N}
        N_j(T) = N_j(s) + \sum_{i\in\FS_j} \lambda_i \int_{s}^T \left( x_{ij}(t - \tau_{ij}) - \ell_j (N_j(t)) \right) dt
        \le N_j(s) + \lambda (T-s)\,,
    \end{align}
    where the first equation follow from \eqref{eq:dynamics-workloads} and the first inequality because $x_{ij}(\cdot) \ge 0$, $\ell_j(\cdot) \le 1$, and $\sum_{i\in\FS_j} \lambda_i \le \sum_{i\in\FS} \lambda_i =: \lambda$. For every $t \in (0,T)$, we have
    \begin{align*}
        \frac 1 T \int_0^T N_j(s) ds &\le \frac 1 T \int_t^T N_j(s) ds \ge \frac 1 T \int_t^T N_j(T) - \lambda (T-s) ds\\
        &= \frac{T-t} T N_j(T) - \lambda \frac {(T-t)^2}{2T}\,,
    \end{align*}
    where the first inequality follows because $N_j(s) \ge 0$, the second inequality from \eqref{eq:lower-bound-N}, and the last equality from integrating and using that $\int_t^T(T-s) dt = (T-t)^2/2$. Re-arranging, we obtain that
    \[
    \frac 1 T {N_j(T)} \ge \frac 1 {T(T-t)} \int_0^T N_j(s) ds + \lambda \frac {T-t}{2T}\,.
    \]
    The result follows from setting $t = T - \sqrt{T}$ and letting $T\rightarrow\infty$.    
\end{proof}

Integrating the dynamics of the backends from $t=0$ to $t=T$, we obtain that
\begin{align*}
    N_j(T) = N_j(0) + \sum_{i\in\FS_j} \lambda_i \int_{0}^T \left( x_{ij}(t - \tau_{ij}) - \ell_j (N_j(t)) \right) dt\,.
\end{align*}
Divide by $T$ and take limits as $T\rightarrow\infty$. Note that $\lim_{T\rightarrow\infty} N_j(T)/T = 0$ by the previous lemma and $\lim_{T\rightarrow\infty} N_j(0)/T = 0$ because initial workloads are finite.
Moreover, we have that
\[
    \frac 1 T \int_{0}^T x_{ij}(t - \tau_{ij}) dt
    = \frac 1 T \int_{0}^T x_{ij}(t) dt + \frac 1 T \int_{-\tau_{ij}}^0 x_{ij}(t) dt - \frac 1 T \int_{T-\tau_{ij}}^T x_{ij}(t) dt\,.
\]
Using that $x_{ij}(\cdot) \in [0,1]$, we obtain by taking limits that $\lim_{T\rightarrow\infty} \frac 1 T \int_{0}^T x_{ij}(t - \tau_{ij}) dt = \bar x_{ij}$. Therefore, putting everything together we obtain that
\[
    \sum_{i\in\FS_j} \lambda_i \bar x_{ij} = \lim_{T\rightarrow\infty} \frac 1 T \int_0^T \ell_j (N_j(t)) dt \le \lim_{T\rightarrow\infty} \ell_j \left( \frac 1 T \int_0^T N_j(t) dt \right) = \ell_j \left( \bar N_j \right)\,,
\]
where the first inequality follows from Jensen's inequality since $\ell_j(\cdot)$ is concave, and the last equality because $\ell_j(\cdot)$ is continuous. The result follows.
\end{proof}

\subsection{Proof of Lemma~\ref{lem:foc}}

\begin{proof}
We first eliminate the variables $\bN$ from the static routing problem. Because the processing rate function is increasing, it is invertible and by flow balance we can write 
\[
    N_j(\bx) = \ell_j^{-1}\left(\sum_{i \in \FS_j} \lambda_i x_{ij}\right)\,.
\]
Therefore, we can succinctly write the static routing problem as
\[
    \operatorname{OPT} = \min_{\bx_i \in \Delta_i} \sum_{j \in \BS} N_j(\bx) + \sum_{(i,j) \in \AS} \lambda_i x_{ij} \tau_{ij}\,.
\]
This is a convex optimization problem because constraints are linear (and thus convex) and the objective is convex. Convexity of the objective follows because, by the inverse function theorem, the derivative of the objective with respect to $x_{ij}$ is $\lambda_i/\ell_j'(\ell_j^{-1}(\sum_{i\in\FS_j} \lambda_i x_{ij})) + \lambda_i \tau_{ij}$, which is increasing because $\ell_i'$ is decreasing and $\ell_i^{-1}$ is increasing. Moreover, strict monotonicity and continuous differentiability of the processing rate functions implies that the objective is continuously differentiable.

Consider the Lagrangian
\begin{align*}
    \mathcal{L}(\bx,c,\nu) = \sum_{j \in \BS} N_j(\bx) + \sum_{(i,j) \in \AS} \lambda_i x_{ij} \tau_{ij} + \sum_{i \in \FS}c_i \left(\lambda_i - \sum_{j\in\BS_i} \lambda_i x_{ij} \right) - \sum_{(i,j) \in \AS} \nu_{ij} x_{ij}\,,
\end{align*}
where $c_i \in \mathbb R$ is the Lagrange multiplier of the flow balance constraint of frontend $i \in \FS$ (pre-multiplied by $\lambda_i$) and $\nu_{ij} \ge 0$ are the Lagrange multipliers of the non-negativity constraints.

Because the problem constraints' are linear and the objective differentiable, the following conditions Karush-Kuhn-Tucker are necessary for optimality:
\begin{align*}
    &\frac{\partial N_j(\bx)}{\partial x_{ij}} + \lambda_i \tau_{ij} - \lambda_i c_i - \nu_{ij} = 0\,, \quad \forall (i,j) \in \AS\,, \\
    &\sum_{j \in \BS_i} x_{ij} = 1\,, \quad \forall i \in \FS\,, \\
    &\nu_{ij} x_{ij} = 0\,, \nu_{ij} \ge 0\,, x_{ij} \ge 0\,, \quad \forall (i,j) \in \AS\,,
\end{align*}
Therefore, from complementary slackness we have that if $x_{ij}^* > 0$ then $\nu_{ij} = 0$. Additionally, Lagrangian stationarity implies that
\[
\frac{\partial N_j(\bx)}{\partial x_{ij}} + \lambda_i \tau_{ij} = \lambda_i c_i\,.
\]
The result follows because, by the implicit function theorem, we have that
\[
    \frac{\partial N_j(\bx)}{\partial x_{ij}} =
    \frac {\lambda_i} {\ell_j'(N_j^*)}\,,
\]
and canceling $\lambda_i$ from both sides. The inequality follows because $\nu_{ij} \ge 0$. Finally, $c_i > 0$ because there always exists a backend $j$ with $x_{ij}^* > 0$ since $\sum_{j \in \BS_i} x_{ij}^* = 1$ together with the fact that $\tau_{ij} \ge 0$ and $\ell_j$ is increasing.
\end{proof}

\subsection{Proof of Lemma~~\ref{lem:existence}}
\begin{proof}
Let $\underline \tau = \min_{(i,j)\in\AS} \tau_{ij}$ be the smallest delay, which is positive by assumption. We use the method of steps, which involves breaking time into intervals $[t_0, t_1], [t_1, t_2], \ldots$ where $t_n = n \underline \tau$. We denote by $(\bN^{(n)}(t), \bx^{(n)}(t))$ the solution in the interval $[t_{n-1}, t_n]$. We prove existence inductively by considering intervals $[t_{n-1},t_n]$ with $n \in \mathbb N$. For the $n$-th interval, we solve the problem
\begin{align}
    \frac {d} {dt} N_j^{(n)}(t) &= \sum_{i\in\FS_j} \lambda_i x_{ij}^{(n-1)}(t - \tau_{ij}) - \ell_j \left(N_j(t)^{(n)}\right) \label{eq:step-N} \\
    \frac {d} {dt} \bx_{i}^{(n)}(t) &= \Pi_{T_{\Delta_i}(\bx_i^{(n)}(t))} \left(-\eta_i \boldsymbol{g}_i^{(n-1)}(t) \right)\,.\label{eq:step-x}
\end{align}
The induction hypotheses is that there exists a unique absolutely continuous solution for $[0,t_{n-1}]$, denoted by $(\bN^{(n-1)}(t), \bx^{(n-1)}(t))$. The method of steps leverages the fact that terms from the previous intervals are fixed from the perspective of the differential equation for the $n$-th period. Once we prove uniqueness and absolutely continuity for the $n$-th period, we construct the solution by ``pasting'' it with the $(n-1)$-th solution.

Note that the differential equations for the workloads and the routing probabilities are decoupled. That is, in \eqref{eq:step-N} the delayed routing probabilities $x_{ij}^{(n-1)}(t - \tau_{ij})$ are fixed data and we only need to solve for $N_j^{(n)}(t)$. The same applies for \eqref{eq:step-x}. Therefore, we can solve each equation at a time.

For the workloads, equation \eqref{eq:step-N} is an ordinary differential equation with a right-hand side that is globally Lipschitz continuous in $N_j$ by Assumption \ref{assume:processing-rate} and continuous in $t$. Therefore, by Picard–Lindel\"of theorem, a unique solution exists.

For the routing probabilities, we write \eqref{eq:step-x} as the following differential inclusion (see, e.g., chapter 3.5 of \citealt{aubincellina1984differential})
\[
\frac {d} {dt} \bx_{i}^{(n)}(t) \in -\eta_i \boldsymbol{g}_i^{(n-1)}(t) - \partial \mathbf 1_{\Delta_i}(\bx_i^{(n)}(t))\,,
\]
where $\mathbf 1_{\mathcal C}(\bx)$ is the indicator function of the set $\mathcal C$ at $\bx$, which is 0 if $\bx \in \mathcal C$ and $\infty$ otherwise and $\partial$ is the subdifferential. Because the set $\Delta_i$ is convex, closed, non-empty, we obtain that $\bx_i \mapsto \mathbf 1_{\Delta_i}(\bx_i)$ is a proper, lower semi-continuous, and convex function~\citep{boyd2004convex}. Therefore, we have that the set-valued differential $\bx_i \mapsto \partial \mathbf 1_{\Delta_i}(\bx_i)$ is maximally monotone by proposition 1 on page 159 from \citet{aubincellina1984differential}. Existence and uniqueness follows from Theorem 3.4 from \citet{brezis1973} because $\boldsymbol{g}_i^{(n-1)}(t)$ is absolutely continuous by Assumption \ref{assume:processing-rate}.
\end{proof}

\subsection{Proof of Lemma~\ref{lem:projection-is-zero}}
\begin{proof}
The first-order optimality conditions of the projection problem imply that $\bv = \Pi_{T_{\Delta_i}(\bx_i)}(-\eta_i \boldsymbol g_i)$ if and only if $(\bv + \eta_i \boldsymbol g_i)^\top (\bz - \bv)\ge 0$ for all $\bz \in T_{\Delta_i}(\bx_i)$. Because zero is optimal, we have that $\boldsymbol g_i^\top \bz \ge 0$ for all $\bz \in T_{\Delta_i}(\bx_i)$.

Because $\bx_i \in \Delta_i$, there always exists some $j \in \BS_i$ such that $x_{ij}>0$. Take $c_i = g_{ij}$. Suppose there exists $j' \neq j$ with $(i,j') \in \AS(\bx)$. Because $x_{ij}, x_{ij'} > 0$, we have that $\boldsymbol e_j - \boldsymbol e_{j'}$ and $\boldsymbol e_{j'} - \boldsymbol e_j$ both lie in $T_{\Delta_i}(\bx_i)$. Therefore, $g_{ij} = g_{ij'}$. For any other $(i,j') \not\in \AS(\bx)$, we have that $\boldsymbol e_{j'} - \boldsymbol e_j \in T_{\Delta_i}(\bx_i)$ because $x_{ij} > 0$, which implies that $g_{ij'} \ge g_{ij}$. The result follows.
\end{proof}

\subsection{Proof of Proposition~\ref{prop:eq-is-opt}}
\begin{proof}
At an equilibrium point $(\bN^{eq},\bx^{eq})$ of ~\eqref{eq:dynamics-workloads}~and~\eqref{eq:dynamics-delay-general} the temporal derivatives must be zero, i.e., $\sum_{i \in \FS_j} \lambda_i x_{ij}^{eq} = \ell_j(N_j^{eq})$ and $\boldsymbol 0 = \Pi_{T_{\Delta_i}(\bx_i^{eq})}(-\eta_i \boldsymbol{g}_i^{eq})$. In the latter, the equilibrium gradients are $g_{ij}^{eq} = 1/\ell_j'(N^{eq}) + \tau_{ij}$ for $(i,j) \in \AS$ and zero otherwise.
By Lemma~\ref{lem:projection-is-zero}, we obtain that there exists some $c_i \in \mathbb R$ such that $g_{ij}^{eq} = c_i$ for $(i,j) \in \mathcal A(\bx^{eq})$ and $g_{ij}^{eq} \ge c_i$ otherwise. Therefore, equilibrium points satisfy the first-order-conditions of \eqref{eq:opt} stated in Lemma~\ref{lem:foc} and thus are stationary points of \eqref{eq:opt}. Under Assumption~\ref{assume:processing-rate}, the problem \eqref{eq:opt} is convex as shown in Lemma~\ref{lem:foc} and, thus, stationary points are globally optimal. 
\end{proof}

\subsection{Proof of Lemma~\ref{lem:projection-analysis}}

\begin{proof}
Let $z_j = -\eta_i g_{ij}$, $\mathcal T = \BS_i(\bx^*)$, $\mathcal S^> = \{ j \in \BS_i \setminus \BS_i(\bx^*) : x_{ij} > 0\}$, and $\mathcal S^0 = \{ j \in \BS_i \setminus \BS_i(\bx^*) : x_{ij} = 0\}$. We ignore all components that have no arcs in the network. The Lagrangian of the projection problem is
\[
    \mathcal L(\bv, \beta, \boldsymbol \mu) = \frac 1 2 \|\bv - \bz\|_2^2 + \beta \mathbf 1^\top \bv - \sum_{j \in \mathcal S^0} \mu_j v_j\,,
\]
where $\beta \in \mathbb R$ is the Lagrange multiplier of the constraint $\sum_{j \in \BS_i} v_j = 0$ and $\mu_j \ge 0$ is the multiplier of the constraint $v_j \ge 0$ for $j \in \mathcal S^0$. The first-order optimality conditions are
\begin{align*}
    \frac{\partial \mathcal{L}}{\partial v_j} = v_j - z_j + \beta - \mu_j \mathbf 1\{ j \in \mathcal S^0 \} = 0\,,
\end{align*}
together with the complementary slackness condition $\mu_j v_j = 0$ for $j \in \mathcal S^0$.

We next solve for $\beta$. Let $\mathcal Q = \mathcal T \cup \mathcal S^> \cup \{ j \in \mathcal S^0: v_j > 0\}$. By feasibility, we have that
\begin{align}\label{eq:sum-of-v-is-zero}
    0 = \sum_{j \in \mathcal T \cup \mathcal S^> \cup \mathcal S^0} v_j = \sum_{j \in \mathcal T \cup \mathcal S^>} v_j +  \sum_{j \in \mathcal S^0 : v_j > 0} v_j =  \sum_{j \in \mathcal Q} \beta - z_j\,,
\end{align}
where we the second equality follows from removing terms $j\in \mathcal S^0$ for which $v_j = 0$, and the third from the first-order conditions together with complementary slackness to obtain that $\mu_j = 0$ for $j \in \mathcal S^0$ with $v_j > 0$. Therefore,
\begin{align}\label{eq:beta-formula}
    \beta = \frac 1 {|\mathcal Q|} \sum_{j \in \mathcal Q} z_j\,.
\end{align}

\paragraph{Part 1} Let $\mathcal P = \mathcal S^> \cup \{ j \in \mathcal S^0: v_j > 0\}$. By definition, we have that $\mathcal Q = \mathcal T \cup \mathcal P$. We have that
\begin{align}\label{eq:bound-sum-v}
    \sum_{j \in \BS_i \setminus \BS_i(\bx^*)} v_j
    &= \sum_{j \in \mathcal P} v_j =
    \sum_{j \in \mathcal P} \left( z_j - \beta \right)
    = \sum_{j \in \mathcal P} \left( z_j - \frac 1 {|\mathcal Q|} \sum_{j' \in \mathcal Q} z_{j'}\right) \nonumber\\
    &=  \frac 1 {|\mathcal Q|} \sum_{j \in \mathcal P} \sum_{j' \in \mathcal Q} \left( z_j - z_{j'} \right)\nonumber\\
    &= \frac 1 {|\mathcal Q|} \sum_{j \in \mathcal P} \sum_{j' \in \mathcal T} \left( z_j - z_{j'} \right) + \frac 1 {|\mathcal Q|} \sum_{j \in \mathcal P} \sum_{j' \in \mathcal P} \left( z_j - z_{j'} \right)\nonumber\\
    &= \frac 1 {|\mathcal Q|} \sum_{j \in \mathcal P} \sum_{j' \in \mathcal T } \left( z_j - z_{j'} \right) \le -\alpha \frac {|\mathcal T| \cdot |\mathcal P|} {|\mathcal Q|} \le  -\frac \alpha 2\,,
\end{align}
where the first equation follows because $\BS_i \setminus \BS_i(\bx^*) = \mathcal S^> \cup \mathcal S^0$ and removing terms $j\in \mathcal S^0$ for which $v_j = 0$, the second equation follows from the first-order conditions, the third equation from \eqref{eq:beta-formula}, the fifth equation because $\mathcal Q = \mathcal T \cup \mathcal P$, the sixth equation because the second summation is zero, the first inequality by the assumption that $z_{j'} > z_{j} + \alpha$ for every $j' \in \mathcal T$ and $j \in \mathcal P$, and the last equality because $|\mathcal P| \ge 1$, $|\mathcal T| \ge 1$ and $|\mathcal Q| = |\mathcal T| + |\mathcal P|$.

\paragraph{Part 2} In this case, we have $\mathcal S^> = \emptyset$. To prove the result we need to argue that $v_j = 0$ for all $j \in \mathcal S^0$. This would imply that $\mathcal Q = \mathcal T$ and, by the first-order conditions, $v_j = z_j - \beta$ for $j \in \mathcal T$.

Suppose there exists $j \in \mathcal S^0$ with $v_j > 0$. Equation~\eqref{eq:bound-sum-v} implies that
\begin{align*}
    \sum_{j \in \mathcal S^0} v_j < 0\,,
\end{align*}
which is a contradiction because all terms are positive.

\paragraph{Part 3} Fix $j \in \BS_i$. If $\mu_j > 0$, then $v_j = 0$ by complementary slackness and the result is trivial. If $\mu_j = 0$, the first-order optimality conditions imply that
\[
  v_j = z_j - \beta = \frac 1 {|\mathcal Q|} \sum_{j \in \mathcal Q} \left( z_j - z_{j'}\right)\,,
\]
where the second equation follows from \eqref{eq:beta-formula}. The result follows from taking absolute values and using that $|a - b| \le \max(a,b)$ for $a,b \ge 0$ together with $|z_j| \le \eta \bar g$.
\end{proof}

\subsection{Proof of Lemma \ref{lem:inactive-stability}}

\begin{proof} By Assumption~\ref{assume:processing-rate} and Assumption~\ref{assume:interior-general} we can find a ball of radius $r>0$ around the optimal workloads $\bN^*$ and constants $\bar g > 0$, $\alpha > 0$, and $L>0$ such that for all $\bN \in \mathbb R^{|\BS|}$ with $|N_j - N_j^*| < r$ we have that:
\begin{enumerate}
  \item Gradients are bounded in a ball around the optimal solution:
\[
 \frac 1{\ell_j'(N_j)} + \tau_{ij} \le \bar g\quad \forall (i,j) \in \AS\,.
\]
  \item Gradients of arcs that are active and inactive at the optimal solution are well separated:
\[
  \frac 1{\ell_{j'}'(N_{j'})} + \tau_{ij'} > \frac 1{\ell_j'(N_j)} + \tau_{ij} + \alpha\quad \forall i \in \FS, j \in \BS_i(\bx^*), j' \in \BS_i \setminus \BS_i(\bx^*)\,.
\]
  \item Processing rate functions are bounded in the ball:
\[
  |\ell(N_j) - \ell(N_j^*)| \le L \quad \forall j\in \BS\,.
\]
\end{enumerate}
We will assume throughout the conditions above hold and then discuss how we can choose $\delta$ small enough to guarantee that $|N_j(t) - N_j^*| < r$.

It suffices to prove the result for $\epsilon \le r$. Consider the Lyapunov function $V(\bx) = \sum_{(i,j) \in \AS \setminus \AS(\bx^*)} x_{ij}$ that gives the flow of arcs that are inactive at the optimal solution. From Lemma~\ref{lem:projection-analysis} we know that if $V(\bx(t)) > 0$ then
\[
  \frac d{dt} V(\bx(t)) = \sum_{(i,j) \in \AS \setminus \AS(\bx^*)} \frac {d}{dt} x_{ij}(t) \le -\frac \alpha 2 \,,
\]
where we used part 1 of Lemma~\ref{lem:projection-analysis} to show that if frontend $i\in\FS$ has an arc with positive flow that is inactive at the optimal solution then  $\sum_{j \in \BS_i \setminus \BS_i(\bx^*)}  \frac {d}{dt} x_{ij}(t) \le - \alpha/2$ and part 2 of Lemma~\ref{lem:projection-analysis} to show that if frontend $i\in\FS$ has no an arc with positive flow that is inactive at the optimal solution then $\sum_{j \in \BS_i \setminus \BS_i(\bx^*)}  \frac {d}{dt} x_{ij}(t)=0$. Let $t_0$ be such that $x_{ij}(t_0) = 0$ for all $(i,j) \in \AS \setminus \AS(\bx^*)$. Integrating we obtain that
\[
  V(\bx(t_0)) - V(\bx(0)) = \int_0^{t_0} \frac d{dt} V(\bx(t)) dt \le -\frac \alpha 2 t_0\,.
\]
Note that $V(\bx(t_0)) = 0$ and
\[
	V(\bx(0)) = \sum_{(i,j) \in \AS \setminus \AS(\bx^*)} \left| x_{ij}(0) - x_{ij}^* \right| \le |\AS| \delta\,,
\]
because inactive arcs have zero flow at the optimal solution (i.e., $x_{ij}^* = 0$) and the initial flows are at most $\delta$ from optimal with respect to the infinity norm. Therefore, we obtain that $t_0 \le 2 |\AS| \delta / \alpha$.

We next bound the growth of the routing probabilities $x_{ij}(t)$ and the workloads $N_j(t)$ up to time $t \in [0,t_0]$. For the routing probabilities, we have that deviations around the equilibrium point satisfy for $t \in [0,t_0]$
\begin{align}\label{eq:growth-of-x}
  \left| x_{ij}(t) - x_{ij}^* \right| &\le \left| x_{ij}(0) - x_{ij}^* \right|  + \int_0^{t_0} \left|\frac {d}{dt} x_{ij}(t)\right| dt \\
  &\le \delta + t_0 \eta \bar g \le \underbrace{\left(1 + \frac{2 \eta \bar g |\AS| }{\alpha} \right)}_{:=C} \delta = C \delta \,,
\end{align}
where the first inequality follows from the triangle inequality and using that $t\le t_0$, the second from part 3 of Lemma~\ref{lem:projection-analysis} together with the fact that gradients are bounded, and the last from our bound on the hitting time $t_0$. For the workloads, we have that deviations around the equilibrium point satisfy
\begin{align*}
    \frac {d} {dt} N_j(t)
    &= \sum_{i\in\FS_j} \lambda_i x_{ij}(t) - \ell_j (N_j(t))\\
    &= \sum_{i\in\FS_j} \lambda_i \left(x_{ij}(t) - x_{ij}^*\right) - \left( \ell_j (N_j(t)) - \ell_j (N_j^*)\right)\,,
\end{align*}
where the second equation follows from flow balance at an equilibrium point as stated in \eqref{eq:equilibrium-N}. Integrating and taking absolute values leads to
\begin{align}\label{eq:growth-of-N}
  \left| N_j(t) - N_j^* \right| &\le \left| N_j(0) - N_j^* \right|  + \int_0^{t_0} \left|\frac {d} {dt} N_j(t)\right| dt \nonumber\\
   &\le \delta +  \sum_{i\in\FS_j} \lambda_i \int_0^{t_0} \left| x_{ij}(t) - x_{ij}^* \right| dt  
   +  \int_0^{t_0} \left| \ell_j (N_j(t)) - \ell_j (N_j^*) \right| dt\nonumber\\
   &\le \delta + \lambda C \delta t_0 + L t_0 \le \left( 1 + \frac{2 (\lambda C \delta + L ) |\AS| }{\alpha}\right) \delta\,,
\end{align}
where the first and second inequality follows from the triangle inequality and using that $t\le t_0$, the third inequality from \eqref{eq:growth-of-x} and using that processing rate functions are bounded and denoting $\lambda = \sum_{i\in\FS_j} \lambda_i$, and the last from the bound on $t_0$.

To conclude, we can use \eqref{eq:growth-of-x} and \eqref{eq:growth-of-N} to pick $\delta > 0$ small enough such that for every $t \in [0,t_0]$ we have $\|(\bN(t),\bx(t)) - (\bN^*, \bx^*) \|_{\infty} < \epsilon$ and $x_{ij}(t_0) = 0$ for all $(i,j) \in \AS \setminus \AS(\bx^*)$.
\end{proof}

\subsection{Proof of Lemma~\ref{lem:spectral-bound}}

\begin{proof}
We first prove a well-known connection between the Dirichet sum and the Laplacian of a graph based on the following algebraic identity:
\[
    \sum_{i=1}^n x_i^2 - \frac 1 n \left(\sum_{i=1}^n x_i\right)^2 = \frac 1 {2n} \sum_{i, j} (x_i - x_j)^2\,.
\]
Using Lemma~\ref{lem:laplacian-matrix2}, we can write the spectral gap as
\[
    \operatorname{gap}\left(\sum_{i\in\FS}\lambda_i \eta_i E_i\right) = \min_{\by \in \mathbb R^{|\BS|} : \mathbf 1^\top \by = 0, \by^\top \by = 1} \frac 1 2  \sum_{i\in\FS} \frac{\lambda_i \eta_i}{|\BS_i|} \sum_{j,j' \in \BS_i} (y_j - y_{j'})^2\,,
\]
where we normalize the vectors $\by$ to be in the unit sphere. Because the feasible set is compact and non-empty, and the objective is continuous, an optimal solution exists by Weierstrass theorem. Let $\by$ be an optimal solution and $j_0 \in \arg\max_{j \in \BS} y_j^2$ be a backend with the highest squared value in the vector $\by$. Because $\mathbf 1^\top \by = 0$, there exists a backend $j_1 \neq j_0$ such that $y_{j_1} y_{j_0} < 0$.

Let $\mathcal P$ be the shortest path between backends $j_0$ and $j_1$. By restricting attention to this path, we can lower bound the objective as follows
\[
    \frac 1 2  \sum_{i\in\FS} \frac{\lambda_i \eta_i}{|\BS_i|} \sum_{j,j' \in \BS_i} (y_j - y_{j'})^2
    \ge \sum_{j\sim i \sim j' \in \mathcal P} (y_j - y_{j'})^2  \frac{\lambda_i \eta_i}{|\BS_i|}\,,
\]
where $j \sim i \sim j'$ denotes a path from backend $j$ to backend $j'$ through frontend $i$. We removed the factor of $1/2$ because every combination of backends is counted twice in the original sum. Alternatively, note that
\begin{align*}
    |y_{j_1} - y_{j_0}| &= \Bigg| \sum_{j \sim j' \in \mathcal{P}} y_{j'} - y_{j} \Bigg|
    \le  \sum_{j \sim i \sim j' \in \mathcal{P}} \left| y_{j'} - y_{j} \right| \cdot \frac{ \sqrt{\lambda_i \eta_i}}{\sqrt{|\BS_i|}} \cdot \frac{ \sqrt{|\BS_i|}}{\sqrt{\lambda_i \eta_i}}\\
    &\le \Bigg(\sum_{j \sim i \sim j' \in \mathcal{P}} \left( y_{j'} - y_{j} \right)^2 \cdot \frac{ \lambda_i \eta_i}{|\BS_i|} \cdot \sum_{i \in \FS(\mathcal{P})} \frac{|\BS_i|}{\lambda_i \eta_i}\Bigg)^{1/2}\,,
\end{align*}
where the first equation follows from telescoping the sum, the first inequality by the triangle inequality, and the last inequality by Cauchy-Schwartz. Therefore, we have
\[
    \sum_{j\sim i \sim j' \in \mathcal P} (y_j - y_{j'})^2  \frac{\lambda_i \eta_i}{|\BS_i|} \ge \frac{(y_{j_1} - y_{j_0})^2}{d(\mathcal P)} \ge \frac{y_{j_0}^2}{d(\mathcal P)} \ge \frac{y_{j_0}^2}{d(\mathcal G)}\,,
\]
where we used our definition for the length of a path, the second inequality follows from expanding the square in the numerator and using that $y_{j_1}^2 \ge 0$ and $y_{j_1} y_{j_0} < 0$, and the last because $d(\mathcal P) \le d(\mathcal G)$. The result follows because $\by$ lies in the unit sphere since
\[
    1 = \sum_{j \in \BS} y_j^2 \le |\BS| \cdot y_{j_0}^2\,,
\]
because $j_0$ is the entry with the highest squared value.
\end{proof}

\subsection{Proof of Lemma~\ref{lem:zero-solution}}
\begin{proof}
When $s=0$, the associated fundamental solution is constant, i.e., $\bar x_{ij}(t) = \bar x_{ij}(0)$ and $\bar N_j(t) = \bar N_j(0)$. We assume without loss that the network $\mathcal G$ is connected. Otherwise, it is sufficient to work with each connected component at a time. Because the solution is constant, equation~\eqref{eq:linear-dynamics-x} implies that for each $(i,j) \in \AS$
\[
\sigma_j \bar N_j(0) = \frac 1 {|\BS_i|}
    \sum_{j' \in \BS_i} \sigma_{j'} \bar N_{j'}(0)\,.
\]
Because the network is connected, we obtain that $\sigma_j \bar N_j(0)$ is constant across backends, i.e., there exists $c \in \mathbb R$ such that $\sigma_j \bar N_j(0) = c$ for all $j \in \BS$. Moreover, equation~\eqref{eq:linear-dynamics-N} implies that for each backend $j \in \BS$
\begin{align}\label{eq:N-at-zero}
\sum_{i\in\FS_j} \lambda_i \bar x_{ij}(0) = \ell_j' \bar N_j(0) = \frac{\ell'_j}{\sigma_j} c\,,
\end{align}
where the last equation follows from the previous observation. Summing over all backends we obtain that
\[
    c \sum_{j \in \BS }\frac{\ell'_j}{\sigma_j} = \sum_{j \in \BS} \sum_{i\in\FS_j} \lambda_i \bar x_{ij}(0) = \sum_{i \in \FS} \lambda_i \sum_{j\in\BS_i}  \bar x_{ij}(0) = 0\,,
\]
where the second equaility follows from exchanging the order of summation, and the last because $\sum_{j \in \BS_i} x_{ij}(t) = 1$ and $\bar{\boldsymbol{x}}_i(t)$ is the difference between two elements in the probability simplex. Because $\ell_j' > 0$ and $\sigma_j > 0$, we obtain that $c = 0$ and it follows that $\bar N_j(0) = 0$ for all $j \in \BS$.

Suppose $\bar x_{ij}(0) \neq 0$ for some $(i,j) \in \AS$. Equation~\eqref{eq:N-at-zero} gives that $\sum_{i\in\FS_j} \lambda_i \bar x_{ij}(0) = 0$ for all $j \in \BS$. Therefore, the solution $x_{ij}^* + \bar x_{ij}(0)$ is feasible and optimal for the static routing problem, which contradicts the uniqueness of the optimal solution from Assumption~\ref{assume:interior-general}. Therefore, $\bar x_{ij}(0) = 0$ for all $(i,j) \in \AS$ and the result follows.
\end{proof}

\subsection{Proof of Lemma~\ref{lem:laplacian-matrix1}}
\begin{proof}
Take $\bx \in \mathbb C^{|\BS|}$. The quadratic form is given by
\begin{align*}
     \bx^\dagger\left(\diag(\ba_i) - \frac{\ba_i \ba_i^\top}{\ba_i^\top \mathbf 1}\right) \bx &= \bx^\dagger \diag(\ba_i) \bx - \frac 1 {|\BS_i|}(\bx^\dagger \ba_i) (\ba_i^\top \bx)
     = \sum_{j \in \BS_i} |x_i|^2 - \frac 1 {|\BS_i|} \Big|\sum_{j\in \BS_i} x_i\Big|^2\,,
\end{align*}
where we used that $\ba_i^\top \mathbf 1 = |\BS_i|$ and $\bx^\dagger \ba_i = \sum_{j\in \BS_i} \overline{x_i} = \overline{\sum_{j\in \BS_i} x_i}$. Positive semi-definiteness follows because
\[
\Big|\sum_{j\in \BS_i} x_i\Big|^2 \le \Big(\sum_{j\in \BS_i} |x_i|\Big)^2
\le |\BS_i| \sum_{j\in \BS_i} |x_i|^2\,,
\]
from the triangle inequality and Jensen's inequality because the quadratic function is convex. The spectral radius is bounded by one because $\Big|\sum_{j\in \BS_i} x_i\Big|^2 \ge 0$ and $\sum_{j \in \BS_i} |x_i|^2 \le \sum_{j \in \BS} |x_i|^2$.
\end{proof}

\subsection{Proof of Lemma~\ref{lem:geometric-bounds}}
\begin{proof}
We perform the change of variables $\tau = \omega / (2 w)$ and $c=\alpha/w$ with $\omega \ge 0$ and $\alpha \ge 0$. With this transformation the function $f$ is given by
\[
    f(\omega) = \frac{ e^{-\im \omega}}{\im \omega (\im \alpha - \im \omega/2 + 1)}\,.
\]

We need to argue that $\Re(f(\omega)) \ge -1 + \alpha \Im(f(\omega))$. We have that
\begin{align*}
    &\Re(f(\omega)) + 1 - \alpha \Im(f(\omega)) =\\
    &\frac {-4\,\sin \left( \omega \right) {\alpha}^{2}+2\,\sin \left(
\omega \right) \alpha\,\omega+4\,{\alpha}^{2}\omega-4\,\alpha\,{\omega
}^{2}+{\omega}^{3}+2\,\cos \left( \omega \right) \omega-4\,\sin
 \left( \omega \right) +4\,\omega}{\omega \left( (2\alpha - \omega)^2 +4 \right) }\,.
\end{align*}
Because the denominator is non-negative, we need to check that the numerator is non-negative. The numerator is convex in $\alpha$ because the coefficient of $\alpha^2$ is $4 (\omega - \sin(\omega)) \ge 0$. Taking derivatives, we obtain that the minimum value is verified at
\[
\alpha^* = \frac {\omega \left( 2\omega - \sin \left( \omega \right)
 \right) }{4(\omega - \sin \left( \omega \right) )} \ge 0\,.
\]
Evaluating the numerator at the minimum value we obtain
\[
{\frac { \left( {\omega}^{2}-16 \right)  \left( \cos \left( \omega
 \right)  \right) ^{2}+8\,\omega\, \left( -\sin \left( \omega \right)
+\omega \right) \cos \left( \omega \right) +15\,{\omega}^{2}-32\,\sin
 \left( \omega \right) \omega+16}{4\omega -4 \sin \left( \omega \right)
\omega}}\,,
\]
which is non-negative for all $\omega \ge 0$.
\end{proof}

\subsection{Proof of Lemma~\ref{lem:exp-bound}}
\begin{proof}
Because $a$ is real we can write $\exp(\im a) = \cos a + \im \sin a$ and, thus
\begin{align*}
    |\exp(\im a) - 1| = \sqrt{(1-\cos a)^2 + \sin^2 a} = \sqrt{2 - 2 \cos a} = \sqrt{2} \sqrt{1 - \cos a} \le |a|\,,
\end{align*}
where the second equation follows because $\cos^2 a + \sin^2 a = 1$ and the inequality because $\cos a \ge 1 - a^2/2$.
\end{proof}

\subsection{Proof of Lemma \ref{lem:spectrum-product-sum}}
\begin{proof}
Consider the matrix $A_\epsilon = A + \epsilon I$, which is invertible for every $\epsilon>0$. Note that $A$ and $A_\epsilon$ have the same eigenvectors and $\operatorname{spec}(A_\epsilon) = \operatorname{spec}(A) + \epsilon$. Using that $A_\epsilon$ is invertible, we obtain that
\begin{align*}
    \operatorname{spec}\left( (A_\epsilon+P) B\right)
    &= \operatorname{spec}\left( A_\epsilon ( B + A_\epsilon^{-1} P B) \right)\subseteq W(A_\epsilon) W\left( B + A_\epsilon^{-1} P B \right)\\
    &\subseteq (\epsilon + W(A)) \left( W(B) + W(A_\epsilon^{-1} P B) \right)\,,
\end{align*}
where the first inclusion follows because $A_\epsilon$ is positive semi-definite, and the second because $W(A_\epsilon) = W(A) + \epsilon$ together with the sub-additivity of the numerical range. The result follows because eigenvalues and the numerical range are continuous functions of $\epsilon$, if $\lim_{\epsilon \downarrow 0} A_\epsilon^{-1} P = A^+ P$. We conclude by proving this last claim.

Because $A$ is positive semi-definite it is diagonalizable. Let $\mu_i \ge 0$ be the eigenvalues and $\by_i$ the associated eigenvectors, which form an orthonormal basis. Denote by $\boldsymbol{P}_j$ the $j$-th column vector of $P$. 
By assumption, for every eigenvector $\by_i$ associated to a zero eigenvalue $\mu_i = 0$, we obtain that $\by_i^\dagger \boldsymbol{P}_j = 0$ since the eigenvector $\by_i$ lies in the kernel of $A$ and $\operatorname{range}(P) \subseteq \operatorname{kernel}(A)^\bot$. Therefore,
\begin{align*}
    A_\epsilon^{-1} \boldsymbol{P}_j = \sum_{i} (\mu_i + \epsilon)^{-1} \by_i \by_i^\dagger \boldsymbol{P}_j
    = \sum_{i : \mu_i > 0} (\mu_i + \epsilon)^{-1} \by_i \by_i^\dagger \boldsymbol{P}_j\,,
\end{align*}
where the first equation follows by  $A_\epsilon^{-1} = \sum_{i} (\mu_i + \epsilon)^{-1} \by_i \by_i^\dagger$ from the formula for the inverse of diagonalizable matrix. The result follows by taking limits and using that the pseudo-inverse is $A^+ =  \sum_{i : \mu_i > 0} \mu_i^{-1} \by_i \by_i^\dagger$.
\end{proof}

\subsection{Proof of Lemma~\ref{lem:laplacian-matrix2}}
\begin{proof}
For the first part, take a vector $\by \in \mathbb C^{|\BS|}$ such that $E_i \by = 0$. Therefore,
\[
\diag(\ba_i)\by = \frac{\ba_i \ba_i^\top}{\ba_i^\top \mathbf 1} \by\,.
\]
This implies that for all $j \in \BS_i$ we have $y_j = 1/|\BS_i| \sum_{j' \in \BS_i} y_{j'}$. Equivalently, we obtain that $y_j = y_{j'}$ for all $j, j' \in \BS_i$.

For the second part, we use that the kernel of a sum of positive semi-definite matrices is the intersection of the kernels to obtain that $\operatorname{kernel}\left(\sum_{i\in\FS} \lambda_i \eta_i E_i\right) = \bigcap_{i\in\FS} \operatorname{kernel}\left(\lambda_i \eta_i E_i\right) = \bigcap_{i\in\FS} \operatorname{kernel}\left(E_i\right)$ where the last equation follows because $\lambda_i > 0$. Because the network is connected, there exists a path connecting every backend, which implies that $y_j = y_{j'}$ for all $j,j' \in \BS$.

The third part follows from the variational definition of eigenvalues and using part 2 to identify the kernel of the matrix.
\end{proof}

\subsection{Proof of Lemma \ref{lem:ratio-bound}}
\begin{proof}
Let $f(w) = (w^2+a^2)/(w^2+b^2)$. Suppose $b\ge a$. We can write $f(w) = 1 - (b^2 - a^2)/(w^2+b^2)$, which is non-decreasing in $w$. Therefore, $f(w) \le \lim_{w \rightarrow \infty} f(w) = 1$ and it follows that $\sqrt{f(w)} \le 1$. Now suppose $a \le b$. In this case $f(w)$ is non-increasing in $w$. Therefore, $f(w) \le f(0) = a^2/b^2$ and it follows that $\sqrt{f(w)} \le a/b$. The result follows from combining both cases.
\end{proof}

\end{document}